\documentclass[12pt]{article}

\usepackage{xcolor}

\usepackage{mathtools}

\usepackage{booktabs}

\usepackage[english]{babel} 

\usepackage[protrusion=true,expansion=true]{microtype} 

\usepackage{amsmath,amsfonts,amsthm}

\usepackage{ amssymb }

\usepackage{graphicx}

\usepackage{array}

\usepackage[margin=1in,hmarginratio=1:1,top=20mm,columnsep=20pt]{geometry} 

\usepackage{booktabs} 

\usepackage{natbib}

\usepackage{setspace}

\usepackage{subcaption}

\usepackage{microtype} 

\usepackage{subfig}

\usepackage{tabularx}

\usepackage[font = small,labelfont=bf,textfont=it]{caption} 

\usepackage{footnote}

\usepackage{algorithm}

\usepackage[noend]{algpseudocode}

\usepackage{enumitem}

\newtheorem{theorem}{Theorem}

\makeatletter

\newcommand{\distas}[1]{\mathbin{\overset{#1}{\kern\z@\sim}}}%

\newsavebox{\mybox}\newsavebox{\mysim}

\newcommand{\distras}[1]{%

  \savebox{\mybox}{\hbox{\kern3pt$\scriptstyle#1$\kern3pt}}%

  \savebox{\mysim}{\hbox{$\sim$}}%

  \mathbin{\overset{#1}{\kern\z@\resizebox{\wd\mybox}{\ht\mysim}{$\sim$}}}%

}

\newcolumntype{C}[1]{>{\centering\let\newline\\\arraybackslash\hspace{0pt}}m{#1}}

\newcommand{\be}{\begin{equation}}

\newcommand{\ee}{\end{equation}}

\newcommand{\bi}{\begin{itemize}}

\newcommand{\ei}{\end{itemize}}

\newcommand{\ben}{\begin{enumerate}}

\newcommand{\een}{\end{enumerate}}

\makeatother

\let\oldbibliography\thebibliography

\renewcommand{\thebibliography}[1]{\oldbibliography{#1}

\setlength{\itemsep}{0pt}} 

\newtheorem{prop}{Proposition}

\title{Exploiting Variance Reduction Potential\\ in Local Gaussian Process Search}

\author{
Chih-Li Sung$^{\rm a}$, 
Robert B. Gramacy$^{\rm b}$, 
Benjamin Haaland$^{\rm a}$
\\ 
$^{\rm a}$Georgia Institute of Technology\\ $^{\rm b}$Virginia Tech
}

\date{}

\begin{document}

\maketitle 

\begin{abstract}
Gaussian process models are commonly used as emulators for computer experiments. 
However, developing a Gaussian process emulator can be computationally prohibitive when the number of experimental samples is even moderately large. 
Local Gaussian process approximation \citep{gramacy2015local} was proposed as an accurate and computationally feasible emulation alternative.
However, constructing local sub-designs specific to predictions at a particular location of interest remains a substantial computational bottleneck to the technique.
In this paper, two computationally efficient neighborhood search limiting techniques are proposed, a maximum distance method and a feature approximation method. Two examples demonstrate that the proposed methods indeed save substantial computation while retaining emulation accuracy. 

\end{abstract}
\section{Introduction}

Due to 
continual advances in computational capabilities,
researchers across fields increasingly rely on computer simulations in lieu of prohibitively costly or infeasible physical experiments. 
One example is \cite{eckstein2013computer}, who use computer simulations to investigate the interaction of energetic particles with solids. 
Physical effects such as elastic energy loss when a particle penetrates a solid, particle transmission through solids, and radiation damage are explored.
These processes can be approximated by simulating the trajectories of all moving particles in a solid based on mathematical models. 
An example in linguistics is the study of language evolution \citep{cangelosi2012simulating}, which is 
made challenging by the unobserved nature of language origin.
Modeling techniques such as genetic algorithms can be used to simulate the process of natural selection and make it possible to explore a virtual evolution. 
While computer simulations provide a feasible alternative to many physical experiments, 
simulating from mathematical models is often itself expensive, in terms of both time and computation, and many researchers seek inexpensive approximations to their computationally demanding computer models---so-called {\em emulators}.

Gaussian process (GP) models \citep{sacks1989design} play an important role as emulators for computationally expensive computer experiments.  They provide an accurate approximation to the relationship between simulation output and untried inputs at a reduced computational cost, and provide appropriate (statistical) measures of predictive uncertainty.
A major challenge in
building a GP emulator for a large-scale computer experiment is that it necessitates decomposing a large ($N\times N$) correlation matrix. 
For dense matrices, this requires around $O(N^3)$ time, where $N$ is the number of experimental runs. Inference for unknown parameters can demand hundreds of such decompositions to evaluate the likelihood, and its derivatives, under different parameter settings for even the simplest Newton-based maximization schemes.  This means that for a computer experiment with as few as $N=10^4$ input-output pairs, accurate GP emulators cannot be constructed without specialized computing resources. 

There are several recent approaches aimed at emulating large-scale computer experiments, most of which focus on approximation of the GP emulator due to the infeasibility of actual GP emulation. 
Examples include covariance tapering which replaces the dense correlation matrix with a sparse version \citep{furrer2006covariance}, multi-step interpolation which successively models global, then more and more local behavior while controlling the number of non-zero entries in the correlation matrix at each stage \citep{haaland2011accurate}, and multiresolution modeling with Wendland's compactly supported basis functions \citep{nychka2014multi}.
Alternatively, \cite{paciorek2013parallelizing} developed an \textsf{R} package called \texttt{bigGP} that combines symmetric-multiprocessors and GPU facilities to handle $N$ as large as $67,275$ without approximation. 
Nevertheless, computer model emulation is meant to avoid expensive computer simulation, not be a major consumer of it. Another approach, proposed by \citet{plumlee2014fast}, is to sample input-output pairs according to a specific design structure, which leads to substantial savings in building a GP emulator. That method, however, can be limited in practice due to the restriction to sparse grid designs.

In this paper, \cite{gramacy2015local}'s local GP approach is considered. The approach is modern, scalable and easy to implement with limited resources. 
The essential idea focuses on approximating the GP emulator at a particular location of interest via a relatively small subset of the original design, thus requiring computation on only a modest subset of the rows and columns of the large ($N\times N$) covariance matrix. 
This process is then repeated across predictive locations of interest, ideally largely in parallel.
The determination of this local subset for each location of interest is crucial since it greatly impacts the accuracy of the corresponding local GP emulator. \cite{gramacy2015local} proposed a greedy search to sequentially augment the subset according to an appropriate criteria and that approach yields reasonably accurate GP emulators. More details are presented in Section \ref{sec:localgaussian}. 

A bottleneck in this approach, however, is that a complete iterative search for the augmenting point requires looping over $O(N)$ data points at each iteration.
In Section \ref{sec:reducedsearch}, motivated by the intuition that there is little potential benefit in including a data point far from the prediction location, two new neighborhood search limiting techniques are proposed, the \textit{maximum distance method} and the \textit{feature approximation method}.
Two examples in Section \ref{sec:example} show that the proposed methods substantially speed up the local GP approach while retaining its accuracy. 
A brief discussion follows in Section \ref{sec:conclusion}. 
Mathematical proofs 
are provided in the Appendix. 


\section{Preliminaries}\label{sec:localgaussian}

\subsection{Gaussian Process Model}
A Gaussian process (GP) is a stochastic process whose finite dimensional distributions are defined via a mean function $\mu(x)$ and a covariance function $\Sigma(x,x')$, for $d$-dimensional inputs $x$ and $x'$. 
In particular,
for $N$ input $x$-values, say $X_N$, which define the $N$-vector $\mu(X_N)$ and $N\times N$ matrix $\Sigma(X_N,X_N)$, and a corresponding $N$-vector of responses $Y_N$, the responses have distribution $Y_N\sim \mathcal{N}(\mu(X_N),\Sigma(X_N,X_N))$. The scale $\sigma^2>0$ is commonly separated from the process correlation function, $Y_N\sim \mathcal{N}(\mu(X_N),\sigma^2\Phi(X_N,X_N))$, where the $N\times N$ matrix $\Phi(X_N,X_N)=(\Phi(x_i,x_j))$ is defined in terms of a correlation function $\Phi(\cdot,\cdot)$, with $\Phi(x,x)=1$. As an example, consider the often-used separable \textit{Gaussian correlation function} 
\begin{equation}\label{eq:guassiancorrelationfunction}
	\Phi_{\Theta}(x,x')=\exp\left\{-\sum^d_{j=1}(x_j-x'_j)^2/\theta_j\right\},
	\mbox{where}\ 
	\Theta=(\theta_1,\ldots,\theta_d),\ \theta_j>0, j=1,\ldots,d.
\end{equation}
Observe that correlation decays exponentially fast in the squared distance between $x_j$ and $x'_j$ at rate $\theta_j$. With this choice, the sample paths are very smooth (infinitely differentiable) and the resulting predictor is an interpolator.

The GP model is popular because inference for $\mu(\cdot),\sigma^2$ and $\Theta$ is easy and prediction is highly accurate. 
A popular inferential choice is maximum likelihood, with corresponding log likelihood (up to an additive constant) 
\begin{gather}
\begin{split}
\ell(\mu,\sigma^2,\Theta)=&-\frac{1}{2}\left\{n\log(\sigma^2)+\log(\det (\Phi_\Theta(X_N,X_N)))+\right.\\
&\qquad\quad\left.(Y_N-\mu(X_N))^T\Phi_\Theta(X_N,X_N)^{-1}(Y_N-\mu(X_N))/\sigma^2\right\}\nonumber
\end{split}
\end{gather}
and the MLEs of $\mu(\cdot),\sigma^2$ and $\Theta$ are 
\begin{equation}\label{eq:mle}
(\hat{\mu}(\cdot),\hat{\sigma}^2,\hat{\Theta})=\arg\max_{\mu,\sigma^2,\Theta}\ell(\mu,\sigma^2,\Theta).
\end{equation}
Here, $\mu(\cdot)$ and its estimate are described somewhat vaguely. Common choices are ${\mu}(\cdot)\equiv 0$, ${\mu}(\cdot)=\mu$, or ${\mu}(\cdot)=h(\cdot)^T\beta$, for a vector of relatively simple basis functions $h(\cdot)$.
More details on inference can be found in \cite{fang2005design} or \cite{santner2013design}.
Importantly, the predictive distribution of $Y(x)$ at a new setting $x$ can be derived for fixed parameters by properties of the conditional multivariate normal distribution. In particular, it can be shown that $Y(x)|X_N,Y_N\sim \mathcal{N}(\mu_N(x),V_N(x))$, where
\begin{equation}\label{emulator}
	\mu_N(x)={\mu}(x)+\Phi_{{\Theta}}(x,X_N)\Phi_{{\Theta}}(X_N,X_N)^{-1}(Y_N-{\mu}(X_N)),
\end{equation}
\begin{equation}\label{varfun}
	V_N(x)=\sigma^2(\Phi_{{\Theta}}(x,x)-\Phi_{{\Theta}}(x,X_N)\Phi_{{\Theta}}(X_N,X_N)^{-1}\Phi_{{\Theta}}(X_N,x)).
\end{equation}
In a practical context, the parameters $\mu(\cdot)$, $\sigma^2$, and $\Theta$ can be replaced by their estimates \eqref{eq:mle} and it might be argued that the corresponding predictive distribution is better approximated by a $t$-distribution than normal
(see 4.1.3 in \cite{santner2013design}). 
Either way,
$\hat\mu_N(x)$ is commonly taken as the \textit{emulator}, and $V_N(x)$ captures uncertainty.

\subsection{Local Gaussian Process Approximation}\label{sec:localGPreview}
A major difficulty in computing the emulator \eqref{emulator} and its predictive variance \eqref{varfun} is solving the linear system $\Phi_{{\hat{\Theta}}}(X_N,X_N)y=\Phi_{\hat{\Theta}}(X_N,x)$, since it requires $O(N^2)$ storage and around $O(N^3)$ computation for dense matrices. 
A promising approach is to search small sub-designs that approximate GP prediction and inference from the original design \citep{gramacy2015local}. The idea of the method is to focus on prediction at a particular generic location, $x$, using a subset of the full data $X_n(x)\subseteq X_N$. Intuitively, the sub-design $X_n(x)$ may be expected to be comprised of $X_N$ close to $x$. 
For typical choices of $\Phi_\Theta(x,x')$, correlation between elements $x$, $x'$ in the input space decays quickly for $x'$ far from $x$, and $x'$'s which are far from $x$ have vanishingly small influence on prediction. Ignoring them in order to work with much smaller, $n\times n$ matrices brings big computational savings, ideally with little impact on accuracy. Figure \ref{fig:localkriging} displays a  smaller sub-design ($n=7$) near location $x=0.5$ extracted from the original design ($N=21$). Although the emulator (red dashed line) performs very poorly from 0 to 0.3 and from 0.6 to 1.0, the sub-design provides accurate and robust prediction at $x=0.5$.

\begin{figure}[h]
\centering
\includegraphics[width=0.6\textwidth]{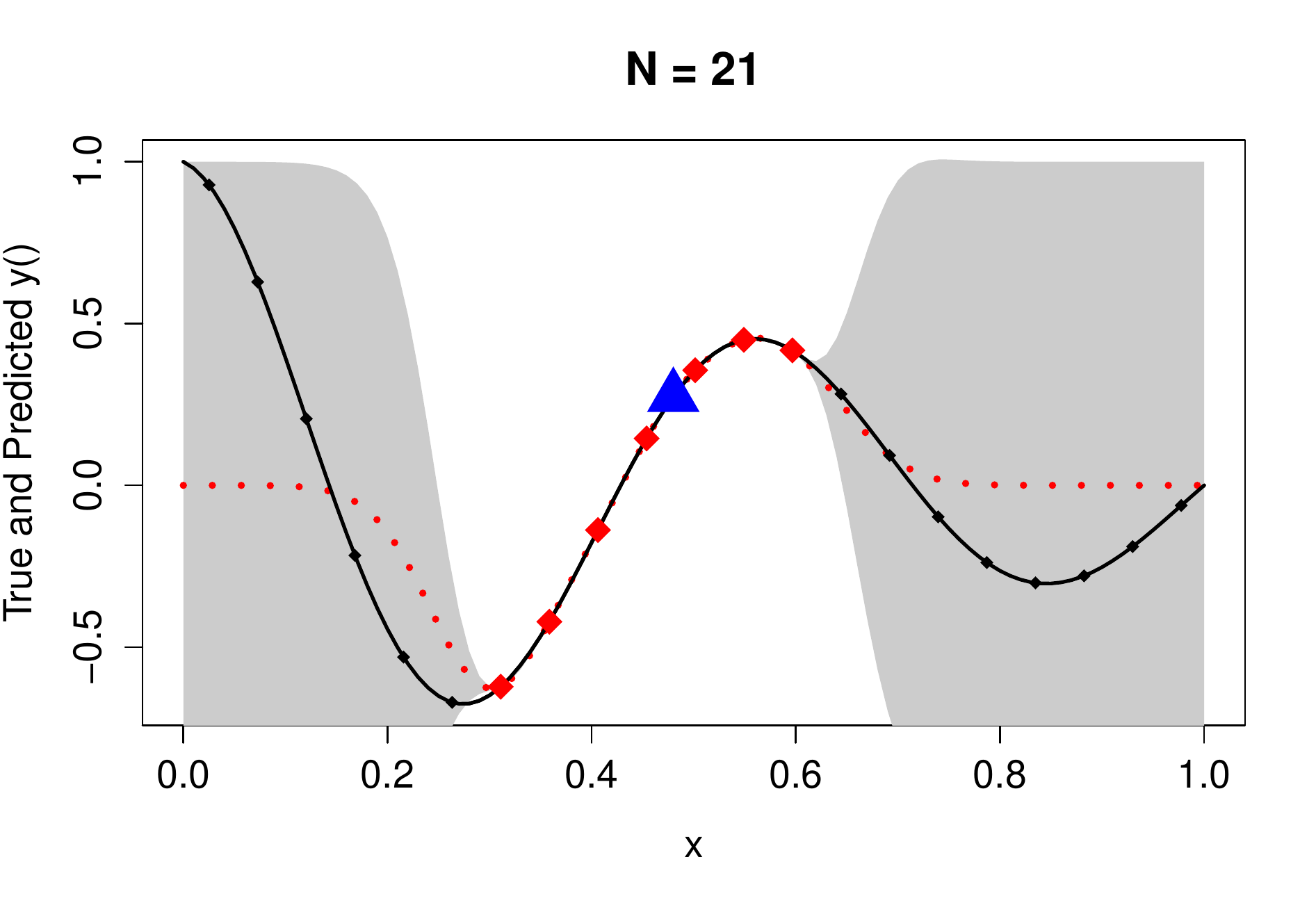}
\caption{An example sub-design $X_7(x)$ for a one dimensional input. Black dots represent the full design, $X_{21}$, the blue triangle represents the point of interest $x=0.5$ and the red diamonds represent the sub-design, $X_7(x)$. Based on the sub-design $X_7(x)$, the emulator is represented as the red dotted line, with the gray shaded region providing a pointwise 95\% confidence band.}
\label{fig:localkriging}
\end{figure}

For an accurate and robust emulator, a smaller predictive variance \eqref{varfun} for each $x$ is desirable. 
We seek a small sub-design $X_n(x)\subseteq X_N$ for each location of interest $x$, which minimizes the predictive variance \eqref{varfun} corresponding to the sub-design $X_n(x)$. This procedure is then repeated for each location of interest $x$. 
The identification of sub-designs and subsequent prediction at each such $x$ can be parallelized immediately, providing a substantial leap in computational scalability.
However, searching for the optimal sub-design, which involves choosing $n$ from $N$ input sites, is a combinatorially huge undertaking. 
A sensible idea is to build up $X_n(x)$ by \textit{n nearest neighbors} (NNs) close to $x$ and the result is a valid probability model for $Y(x)|X_n(x),Y(X_n(x))$ \citep{datta:etal:2015}. 
\cite{gramacy2015local} proposed a greedy, iterative search for the sub-design, starting from a small NN set $X_{n_0}$ and sequentially choosing the $x_{j+1}$ which provides the greatest reduction in predictive variance to augment $X_j(x)$, for $j=n_0,n_0+1,\ldots,n$. That is, 
\begin{equation}\label{eq:minimunvariance}
x_{j+1}=\underset{\begin{subarray}{c}u\in X_N\setminus X_j(x),\\X_{j+1}=X_j(x)\cup u
\end{subarray}}{\arg\min}V_{j+1}(x)
\end{equation}
and $X_{j+1}(x)=X_j(x)\cup x_{j+1}$.
Both the greedy and NN 
schemes can be shown to have computational order $O(n^3)$ (for fixed $N$) when the scheme is efficiently deployed for each update $j\rightarrow j+1$. 
Specifically, the matrix inverse $\Phi_{\Theta}(X_{j+1},X_{j+1})^{-1}$ in $V_{j+1}(x)$ can be updated efficiently using partitioned inverse equations \citep{harville1997matrix}. 
Before the greedy subsample selection proceeds, correlation parameters can be 
initialized to reasonable fixed values 
to be used throughout the sub-design search iterations.  
After a sub-design has been selected for a particular location, 
a local MLE can be constructed.
Thus, only $O(n^3)$ cost is incurred for building the local subset and subsequent local parameter estimation.  For details and implementation, see 
the \texttt{laGP} package for \textsf{R} \citep{lagpjss}. 
An initial \emph{overall} estimate of the correlation parameters can be obtained using 
the Latin hypercube design-based block bootstrap subsampling scheme proposed by \cite{liu:hung:2015}, which has been shown to
consistently estimate \emph{overall} lengthscale $\theta_j$-values in a computationally tractable way, even with large $N$.

The greedy scheme, searching for the next design point in $X_N\setminus X_j(x)$ to minimize the predictive variance \eqref{eq:minimunvariance}, is still
computationally expensive, especially when the design size $N$ is very large. 
For example, the new $x_{j+1}$  based on \eqref{eq:minimunvariance} involves searching over $N-j$ candidates. 
In that case, the greedy search method  
still contains a serious computational bottleneck
in spite of its improvements 
relative to solving the linear system in \eqref{emulator} for GP prediction and inference. \cite{gramacy2014massively} recognized this issue and accelerated the search by exporting computation to graphical processing units (GPUs). They showed that the GPU scheme with local GP approximation and massive parallelization can lead to an accurate GP emulator for a one million run full design, with the GPUs providing approximately an order of magnitude speed increase. 
\cite{gramacy2015speeding} noticed that the progression of $x_{j+1},j=1,2,\ldots$ \emph{qualitatively} takes on a ribbon and ring pattern in the input space and suggested a computationally efficient heuristic based on one dimensional searches along rays emanating from the predictive location of interest $x$.

In Section \ref{sec:reducedsearch}, two computationally efficient and accuracy preserving neighborhood search methods are proposed.
Both neighborhood searches reduce computation by decreasing the number of candidate design points examined.
It is shown that only locations within a particular distance of either the prediction location $x$ or the current sub-design, or locations in particular regions within a \emph{feature space}, can have {\it substantial} influence on prediction. 
Using these techniques, it is possible 
to search a \emph{much} smaller candidate set at each stage, 
leading to huge reductions in computation and increases in scalability. 

\section{Reduced Search in Local Gaussian Process}\label{sec:reducedsearch}
As discussed previously, when building a sub-design $X_n(x)$ for prediction at location $x$, there is intuitively little potential benefit to considering input locations which are very distant from $x$ (relative to the correlation decay) as the response value at these locations is nearly independent of the response at $x$. 
In Section \ref{sec:method1},
a \textit{maximum distance} bound and corresponding algorithm are provided 
and in Section \ref{sec:method2}, a \textit{feature approximation} bound and corresponding modification to the algorithm are provided. 
The algorithms
furnish a dramatically 
reduced set of potential design locations which need to be examined, in a computationally efficient and scalable manner.

\subsection{Maximum Distance Method}\label{sec:method1}
Following the notation from Section \ref{sec:localgaussian}, $x$ is the particular location of interest, in terms of emulation/prediction, and $X_j(x)$ is the greedy sub-design at stage $j$. 
To augment the sub-design $X_j(x)$, the locations which are distant from $x$ should intuitively have little potential to reduce the predictive variance at $x$.  Therefore their consideration as potential $x_{j+1}$ values
is unnecessary.
This intuition is correct and developed as follows.

First, assume that the underlying correlation function is radially decreasing after appropriate linear transformation of the inputs. That is, assume there is a strictly decreasing function $\phi$ so that $\Phi_\Theta(x,x')=\phi(\|\Theta(x-x')\|_2)$ for some $\Theta$. 
In practice, $\Theta$ can be estimated using the local MLE as discussed in Section \ref{sec:localGPreview}, using as a starting value the \emph{overall}, consistent estimate from the sub-design search iterations. Now, consider a candidate input location $x_{j+1}$ at stage $j+1$ of the greedy sub-design search for an input location to add to the design and define $d_{\min}(x_{j+1})$ as the minimum (Mahalanobis-like) distance between the candidate point $x_{j+1}$ and the current design and location of interest, that is, 
\begin{align}\label{mindisdef}
\nonumber
d_{\min}(x_{j+1})=\min\{\|\Theta(x-&x_{j+1})\|_2,\|\Theta(x_1-x_{j+1})\|_2,\\
&\|\Theta(x_2-x_{j+1})\|_2,\ldots,\|\Theta(x_j-x_{j+1})\|_2\}.
\end{align}
For example, consider the sub-design $X_j(x)$ with two dimensional inputs shown in Figure \ref{mindistance} for $j=8$. The location of interest is marked with a circled $\times$ and the current sub-design $X_j(x)$ is indicated with gray dots. With $\Theta={\rm diag}(1/\sqrt{3},1/\sqrt{3})$, the candidate points $x_{j+1}$ with $d_{\min}(x_{j+1})$ less than $3.07$ lie within the yellow shaded region.

\begin{figure}[h]
\centering
\includegraphics[width=0.5\textwidth]{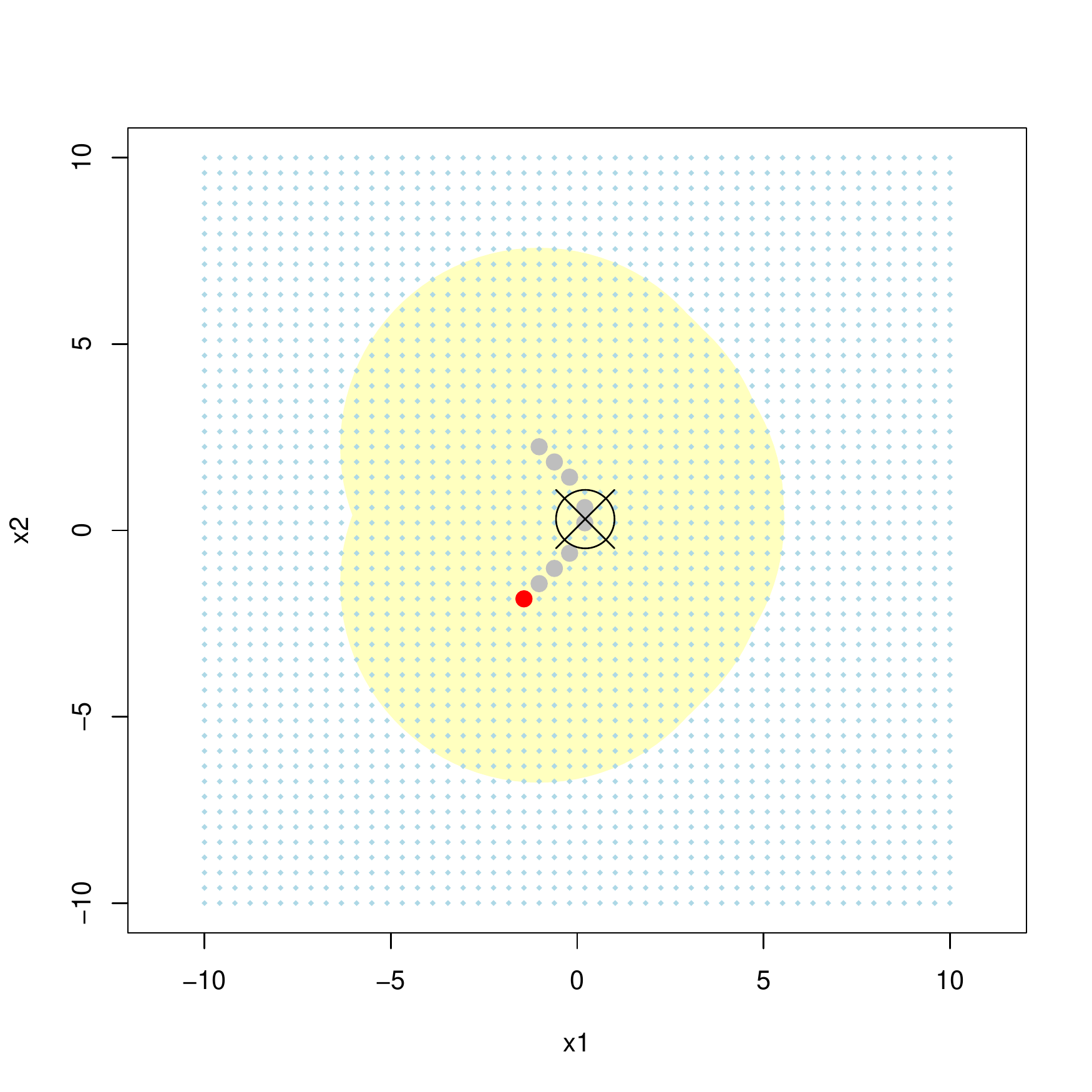}
\caption{An example sub-design $X_{8}(x)$ with two dimensional inputs. The circled $\times$ represents the location of interest. With $\Theta={\rm diag}(1/\sqrt{3},1/\sqrt{3})$, the gray dots represent current design points $X_{8}(x)$, the red dot represents the new input location $x_{9}$, and the yellow shaded region represents the candidate points $x_*$ with $d_{\min}(x_*)<3.07$.}
\label{mindistance}
\end{figure}

Based on the local design scheme introduced in Section \ref{sec:localgaussian} and equation \eqref{eq:minimunvariance}, the sub-design $X_n(x)$ is built up through the choices of $x_{j+1}$ to sequentially augment $X_j(x)$, at each stage aiming to minimize predictive variance. Proposition \ref{prop1} provides an alternate formula for this variance, which will be used to greatly reduce the number of candidates in the minimization problem. Its proof is provided in Appendix \ref{proofreducevariance}.

\begin{prop}\label{prop1}
The predictive variance $V_j(x)$ in \eqref{varfun} can be represented via the recurrence 
\begin{equation}\label{eq:recurrent}
V_{j+1}(x)=V_j(x)-\sigma^2R(x_{j+1}).
\end{equation}
Here, $R(x_{j+1})$ represents the (scaled) \textit{reduction in variance}. In particular,
\begin{equation}\label{reducevariance}
R(x_{j+1})=\frac{(\Phi_\Theta(x,x_{j+1})-\Phi_\Theta(x_{j+1},X_j)\Phi_\Theta(X_j,X_j)^{-1}\Phi_\Theta(X_j,x))^2}{\Phi_\Theta(x_{j+1},x_{j+1})-\Phi_\Theta(x_{j+1},X_j)\Phi_\Theta(X_j,X_j)^{-1}\Phi_\Theta(X_j,x_{j+1})}.
\end{equation}
\end{prop}


The recurrence relation \eqref{eq:recurrent} is useful for searching candidates to entertain. Further, minimizing variance after adding the new input location $x_{j+1}$ is equivalent to maximizing reduction in variance $R(x_{j+1})$. 

The following theorem allows one to narrow the window of candidate locations to consider when searching greedily for a local design. The proof is provided in Appendix \ref{proofthm1}.

\begin{theorem}\label{theorem}
Suppose $\Phi:\Omega\times\Omega\rightarrow\mathbb{R}$ is a symmetric positive-definite kernel on a compact set $\Omega\subseteq \mathbb{R}^d$ and there exists a strictly decreasing function $\phi:\mathbb{R}^+\rightarrow\mathbb{R}$ such that $\Phi_\Theta(x,y)=\phi(\|\Theta(x-y)\|_2)$ for some $\Theta$. 
Then, for $\delta>0$, $R(x_{j+1})\leq \delta$
if 
\begin{equation}\label{inequation}
d_{\min}(x_{j+1})\geq\phi^{-1}\left(\sqrt{\frac{\delta}{(1+\sqrt{j}\|\Phi_\Theta(X_j,X_j)^{-1}\Phi_\Theta(X_j,x)\|_2)^2+j\delta/\lambda_{\min}}}\;\right),
\end{equation}
where $\lambda_{\min}$ is the minimum eigenvalue of $\Phi_\Theta(X_j,X_j)$. 
\end{theorem}

This result indicates that candidate locations which are sufficiently distant from both the location of interest and the current sub-design do not have potential to reduce the variance more than $\delta$. Importantly, if the full set of design locations $X_N$ is stored in a data structure such as a {\em k-d tree} \citep{bentley1975multidimensional}, then the set of candidate locations which do not satisfy inequality \eqref{inequation} can be identified in $O(\log N)$ time, with constant depending on $\delta$, dimension of the input space, and stage $j$, which provides a computationally efficient and readily scalable 
technique for reducing the set of
potential candidate locations.

Theorem \ref{theorem} suggests Algorithm \ref{alg:algorithm1} as a starting point for efficiently selecting sub-designs for prediction at location $x$. In the algorithm, a larger value of $\delta$ is desirable since larger $\delta$ leads to fewer candidate design locations to search. One way to obtain a relatively large value of $\delta$ is to examine the variance reductions on the set of $k$ nearest neighbors which are not yet in the sub-design, which is shown in Step \ref{alg1:step2}. 
The number of nearest neighbors $k$ is a tuning parameter. A larger value of $k$ will provide a larger variance reduction and therefore exclude more candidate design locations, albeit at an additional computational expense since the variance reduction must be checked at each of these locations.
Alternatively, a large value of $\delta$ could be obtained by applying the heuristic proposed in \cite{gramacy2015speeding}.  From the result of Theorem \ref{theorem}, $T(X_j)$ in Step \ref{alg1:step3}, which indicates the region such that 
\begin{gather}
d_{\min}(x_{j+1})\leq\phi^{-1}\left(\sqrt{\frac{\delta}{(1+\sqrt{j}\|\Phi_\Theta(X_j,X_j)^{-1}\Phi_\Theta(X_j,x)\|_2)^2+j\delta/\lambda_{\min}}}\;\right),\label{y}
\end{gather}
gives the subset of candidate locations that have potential to reduce the variance more than $\delta$. 

For each update $j\rightarrow j+1$, the algorithm involves $O(j^2+j\log N)$ computation in Step \ref{alg1:step3}, $O(j\log N)$ for eliminating search locations and $O(j^2)$ for computing
the right-hand side of (\ref{y}), 
the maximum distance from the current design and location of interest.
In particular, the matrix inverse $\Phi_{\Theta}(X_j,X_j)^{-1}$ can be updated via the partitioned inverse equations \citep{harville1997matrix} with $O(j^2)$ cost at each iteration. Analysis of the computational complexity of obtaining (an approximation to) the minimum eigenvalue of $\Phi_\Theta(X_j,X_j)$ is more challenging.
It is convenient to work with the reciprocal of the maximum eigenvalue of 
$\Phi_\Theta(X_j,X_j)^{-1}$, for which relatively efficient algorithms such as the power or Lanczos method exist
\citep{golub1996matrix}.
If the starting vector is not orthogonal to the target eigenvector, then convergence of the (less efficient, but easier to analyze) power method is geometric with rate depending on the ratio between the two largest eigenvalues of $\Phi_\Theta(X_j,X_j)^{-1}$ (see equation 9.1.5 in \cite{golub1996matrix}). 
While this rate and the constants in front are not fixed across $j$, they can be bounded, with the exception of the influence of the starting vector, across all subsets of the full dataset. The starting vector might be expected to be increasingly collinear with the target eigenvector as $j$ increases, thereby improving the rate bound. 
All together this implies an approximately constant number of iterations, each costing $O(j^2)$, is required to approximate $\lambda_{\rm min}$ for each $j$.
Another perspective would be to choose a random starting vector, for which \cite{kuczynski1992estimating} provide respective average and probabilistic bounds of $O(j^2\log j)$ for the power method and $O(j^2\log^2 j)$ for the Lanczos method.
The inverse function $\phi^{-1}:\mathbb{R}\rightarrow\mathbb{R}$ can be computed in roughly constant time by a root-finding algorithm or even computed exactly for many choices of $\Phi$.
For example, consider the power correlation function, i.e., $\Phi_{\Theta}(x,y)=\exp\{-\|\Theta(x-y)\|_2^p \}$, the $\phi$ can be formed as $\phi(u)=\exp\{-u^p\}$, so $\phi^{-1}(v)=(-\log v)^{1/p}$.
Note that when a large $n$ is required, computation of $\lambda_{\min}$ might be numerically unstable. 
A remedy in that case may be to stop the search when $\lambda_{\min}$ falls below a prespecified threshold or 
perhaps introduce a penalty inversely proportional to $\lambda_{\min}$.

\begin{algorithm}[htb]
  \caption{ Maximum distance search method in local Gaussian process.}
  \label{alg:algorithm1}
  \begin{algorithmic}[1]
          \State Set $j=1$ and $x_1$ as the point closest to the predictive location $x$. Throughout, let $X_j(x) \equiv X_j=\{x_1,x_2,\ldots,x_j\}$, dropping the explicit $(x)$ argument.
          \State Let $N_{jk}(x)$ denote the $k$ nearest neighbors to $x$ in $X_N\setminus X_j$, the candidate locations not currently in the sub-design. Set $\delta_{j+1}$ equal to the maximum variance reduction from $N_{jk}(x)$. That is,
          \begin{equation}\label{threshold}
          \delta_{j+1}=\max\limits_{u \in N_{jk}(x)}{R(u)},
          \end{equation}
          where $R(\cdot)$ is shown in \eqref{reducevariance}.
          \label{alg1:step2}
          \State Set $y=\phi^{-1}\left(\sqrt{\frac{\delta_{j+1}}{(1+\sqrt{j}\|\Phi_\Theta(X_j,X_j)^{-1}\Phi_\Theta(X_j,x)\|_2)^2+j\delta_{j+1}/\lambda_{\min}}}\right)$, where $\Phi_\Theta(x,x')=\phi(\|\Theta(x-x')\|_2)$ and $\lambda_{\min}$ is the minimum eigenvalue of $\Phi_\Theta(X_j,X_j)$.
Let 
          \begin{equation}\label{method1:candidate}
          T(X_j)=\{u\in X_N\setminus X_j:\|\Theta(u-v)\|_2\leq y \text{ for some } v\in \{x,X_j\}\}.
          \end{equation} Then,
          \begin{displaymath}
           x_{j+1}=\arg\max\limits_{u\in T(X_j)}{R(u)}.
          \end{displaymath}
          \label{alg1:step3}
          \State Set $j=j+1$ and repeat \ref{alg1:step2} and \ref{alg1:step3} until either the reduction in variance $R(x_{j+1})$ falls below a prespecified threshold or the local design budget is met.
  \end{algorithmic}
\end{algorithm}

\subsection{Feature Approximation Method}\label{sec:method2}
In addition to the maximum distance method and associated algorithm, an approximation via eigen-decomposition can be applied to reduce the potential locations in a computationally efficient manner. 
Suppose that $\Phi$ is a symmetric positive-definite kernel on a compact set $\Omega\subseteq \mathbb{R}^d$ and $P:L_2(\Omega)\rightarrow L_2(\Omega)$ is an integral operator, defined by 
\begin{equation}\label{defineT}
Pv(x):=\int_{\Omega}\Phi(x,y)v(y)dy, \quad v\in L_2(\Omega), \quad x\in\Omega.
\end{equation}
Then, Mercer's theorem guarantees the existence of a countable set of positive eigenvalues $\{\lambda_j\}^\infty_{j=1}$ and an orthonormal set $\{\varphi_j\}^\infty_{j=1}$ in $L_2(\Omega)$ consisting of the corresponding eigenfunctions of $P$, that is, $P\varphi_j=\lambda_j\varphi_j$ \citep{wendland2004scattered}. Furthermore, the eigenfunctions $\varphi$'s are continuous on $\Omega$ and $\Phi$ has the absolutely and uniformly convergent representation
\begin{equation*}
\Phi(x,y)=\sum^\infty_{j=1}{\lambda_j\varphi_j(x)\varphi_j(y)}.
\end{equation*}

In particular, $\Phi$ can be approximated uniformly over inputs in terms of a finite set of eigenfunctions
\begin{equation}\label{approx_eigen_decomp}
\Phi(x,y)\approx\sum^D_{j=1}{\lambda_j\varphi_j(x)\varphi_j(y)}
\end{equation} for some moderately large integer $D$. For some kernel functions, closed form expressions exist. For example, the Gaussian correlation function \eqref{eq:guassiancorrelationfunction} (on $\mathbb{R}^d$, with weighted integral operator) has eigenfunctions given by products of Gaussian correlations and Hermite polynomials \citep{zhu1997gaussian}. More generally, \cite{williams2001using} show high-quality approximations to these eigen-decompositions can be obtained via Nystr{\"o}m's method.

\begin{theorem}\label{featureThm}
Assume $\Phi:\Omega\times\Omega\rightarrow \mathbb{R}$ is a symmetric positive-definite kernel on a compact set $\Omega\subseteq \mathbb{R}^d$ which can be approximated via $D$ eigenfunctions (see equation \eqref{approx_eigen_decomp}). Then, the reduction in variance \eqref{reducevariance} has approximate representation
\begin{equation}\label{reducevariancebyfeature}
	R(x_{j+1})\approx\|C_{X_j}(x)\|^2_2\cos^2(\vartheta),
\end{equation} 
where $\vartheta$ is the angle between $C_{X_j}(x)$ and $C_{X_j}(x_{j+1})$,  
\begin{align}
	C_{X_j}(t)&=[I-U(X_j)[U^T(X_j)U(X_j)]^{-}U^T(X_j)]U(t),\label{defineC}\\
	U(t)&=\left(\sqrt{\lambda_1}\varphi_1(t),\ldots,\sqrt{\lambda_D}\varphi_D(t)\right)^T,\quad{\rm and}\nonumber\\
U(X_j)&=\left[U(x_1),\ldots,U(x_j)\right], \nonumber
\end{align}
 for eigenfunctions $\varphi_i(t)$ and corresponding ordered eigenvalues $\lambda_1\ge\ldots\ge\lambda_D$.
\end{theorem}
\begin{proof}
Provided in Appendix \ref{proof:reducevariancebyfeature}.
\end{proof}

According to this approximation, instead of excluding candidates in Euclidean space as indicated in Theorem \ref{theorem}, the candidate set can be further reduced by transforming the
inputs into a feature space. A modified algorithm is suggested as follows. The variance reduction threshold in equation \eqref{threshold} now places a restriction on the \textit{angle} between $C_{X_j}(x)$ and $C_{X_j}(x_{j+1})$, where we would like to exclude points \emph{outside} the cones
\begin{equation}\label{subregion}
\cos^2(\vartheta)\leq\frac{\delta_{j+1}}{\|C_{X_j}(x)\|^2_2}.
\end{equation}

A feature approximation modification to Algorithm \ref{alg:algorithm1} is shown in Algorithm \ref{alg:algorithm2}.
%
%
To reduce the computational burden in checking \eqref{subregion}, the values of the first $D$ eigenfunctions at the full dataset $X_N$, $U(X_N)$, could be computed in advance and stored based on a locality-sensitive hashing (LSH) scheme \citep{indyk1998approximate}. LSH is a method for answering approximate similarity-search queries in high-dimensional spaces. The basic idea is to use special locality-sensitive functions to \textit{hash} points into ``buckets" such that ``nearby" points map to the same bucket with high probability. Many similarity measures have corresponding LSH functions that achieve this property. For instance, the hashing functions for cosine-similarity are the normal vectors of random hyperplanes through the origin, denoted for example as $v_1,\ldots,v_k$. Depending on its side of these random hyperplanes, a point $p$ is placed in bucket $h_1(p),\ldots,h_k(p)$, where $h_i(p)=\text{sign}(v_i^Tp)$. A simple example, following \cite{van2010online}, is provided in Figure \ref{fig:LSH_illustration}. Figure \ref{fig:LSH_illustration_a} illustrates the hashing process for a point $p$, where the point $p$ is hashed into the bucket $(h_1(p),\ldots,h_6(p))=(-1,-1,1,1,1,1)$ by the definition $h_i(p)=\text{sign}(v_i^Tp),i=1,\ldots,6$ (when the point $p$ is above the hyperplane, the inner product is negative, otherwise the inner product is positive). 
Similarly, other points are placed in their corresponding buckets. 
In the search process, shown in Figure \ref{fig:LSH_illustration_b}, the query point $q$ is mapped to the bucket $(h_1(q),\ldots,h_6(q))=(-1,1,1,1,1,1)$, which matches the bucket of point $p'$. Thus, the hashing and search processes retrieve $p'$ as the \emph{most similar} neighbor of $q$. Also, since the one different label in the buckets of $p$ and $q$ implies that the angular difference is close to $\pi/6$ (six hyperplanes), $p$ is retrieved when querying the points whose angular difference from $q$ is less than $\pi/6$. Note that many more than six hyperplanes are 
needed to ensure that the
returned angle similarity is approximately correct.
In a standard LSH scheme, the hashing process is performed several times by different sets of random hyperplanes, and the search procedure iterates over these random sets of hyperplanes. More details and examples can be seen in \cite{indyk1998approximate},\cite{van2010online}, and \cite{leskovec2014mining}.

\begin{figure}
\centering
\begin{subfigure}{.45\textwidth}
\centering
\vspace{-0.1cm}
\includegraphics[width=.9\linewidth]{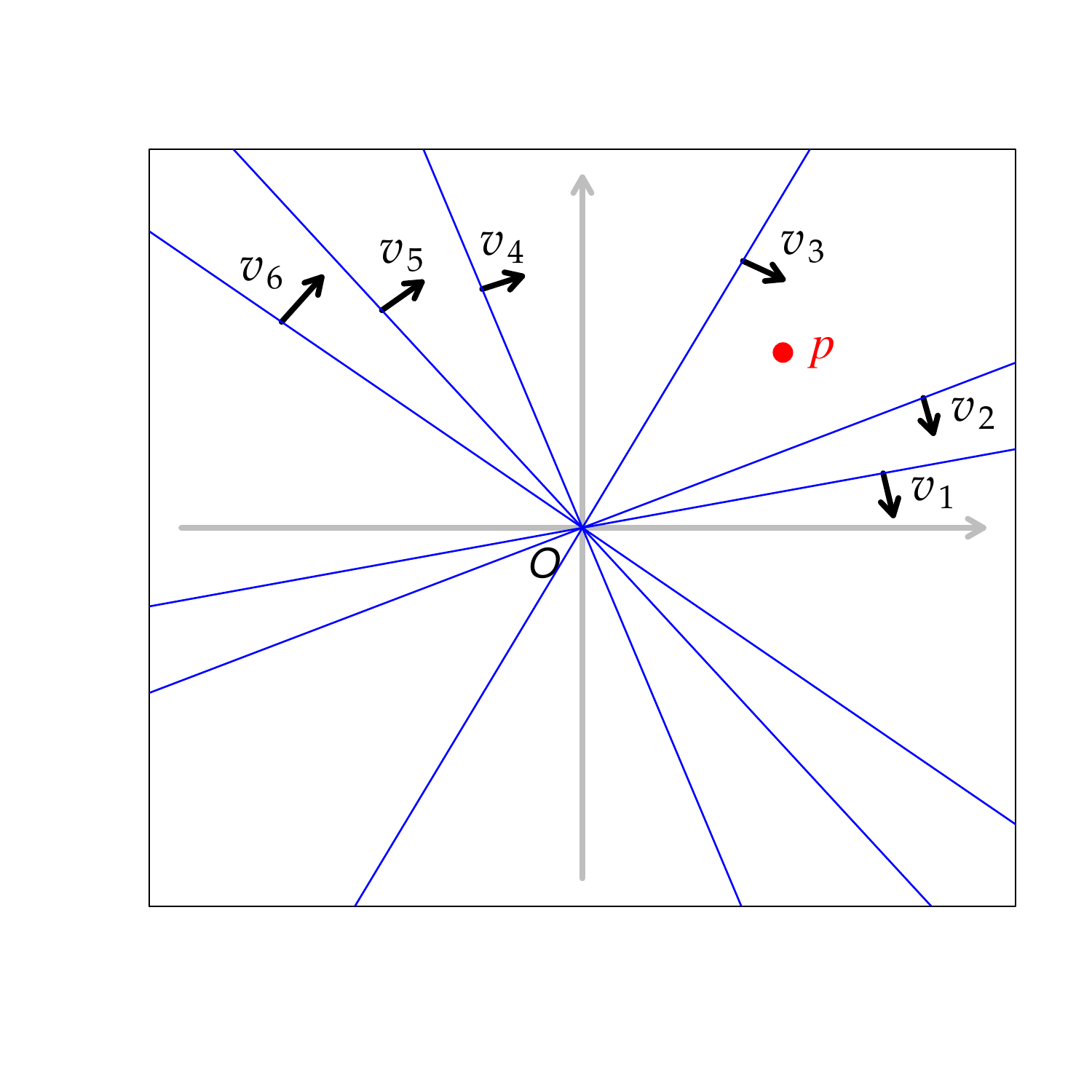}
\vspace{-0.65cm}
\caption{Hashing.}
\label{fig:LSH_illustration_a}
\end{subfigure}
\begin{subfigure}{.45\textwidth}
\centering
\vspace{-0.1cm}
\includegraphics[width=.9\linewidth]{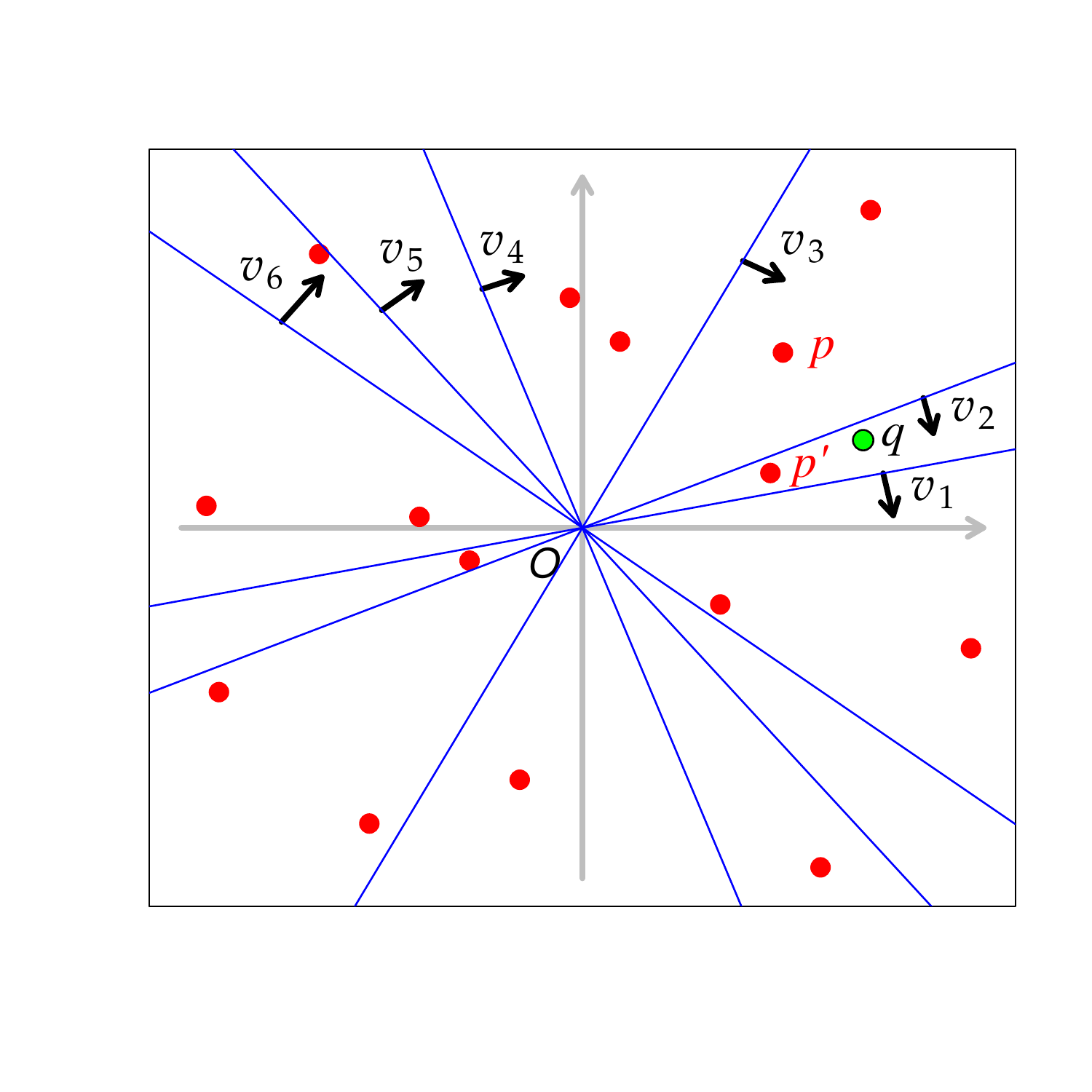}
\vspace{-0.65cm}
\caption{Search.}
\label{fig:LSH_illustration_b}
\end{subfigure}
\caption{\small Illustration of locality-sensitive hashing (LSH) scheme. Blue lines are random hyperplanes through origin, and $v_1,\ldots,v_6$ (black arrows) are the corresponding normal vectors. Red dots present stored data points, and green dot presents the query data point.}
\label{fig:LSH_illustration}
\end{figure}

Apart from cosine-similarity, \cite{jain2008fast} showed for the pairwise similarity 
\[
\frac{y_k^TA_jy_h}{\|G_jy_k\|_2\|G_jy_h\|_2},
\] where $y_k,y_h\in \mathbb{R}^d$, $G_j^TG_j=A_j$ and $A_j$ is a $d\times d$ positive-definite matrix that is updated for each iteration $j$, the hash function can be defined as:
\begin{equation}\label{hashfunction}
	h_{A_j}(y)=\begin{cases} 
      1 & r^TG_jy\geq 0 \\
      0 & \text{otherwise}
   \end{cases} ,
\end{equation}
where the vector $r$ is chosen at random from a $d$-dimensional Gaussian distribution. Let $G_j=I-U(X_j)[U^T(X_j)U(X_j)]^{-}U^T(X_j)$ and $A_j=G_j$ ($G_j$ is symmetric and idempotent), then $\cos(\vartheta)$ in \eqref{reducevariancebyfeature} can be represented as
\[
\cos(\vartheta)=\frac{U(x_{j+1})^TA_jU(x)}{\|G_jU(x_{j+1})\|_2\|G_jU(x)\|_2}.
\]
Thus, in the feature approximation method, an LSH scheme can be employed by storing $U(X_N)$ in advance and updating the hash function \eqref{hashfunction} at each iteration, where $y$ is replaced by $U(y)$. At query time, similar points are hashed to the same bucket with the query $U(x)$ and the results are guaranteed to have a similarity within a small error after repeating the procedure several times. In particular, for each update $j\rightarrow j+1$, given that the LSH method guarantees retrieval of points within the radius $(1+\epsilon)M$ from the query point $U(x)$, where $M$ is the distance of the true nearest neighbor from $U(x)$, the method requires $O(D^2+jDN^{1/(1+\epsilon)})$ computational cost, $O(D^2)$ for updating matrix $G_j$ (via the partitioned inverse equations \citep{harville1997matrix}) and computing the hash function $h_{A_j}(y)$ (via the implicit update in \cite{jain2008fast}), and $O(jDN^{1/(1+\epsilon)})$ for identifying the hashed query \citep{jain2008fast}, where $D$ is the number of eigenfunctions in Theorem \ref{featureThm}. In Section \ref{sec:example}, two examples show the benefit from the LSH approach in the feature approximation method.

\begin{algorithm*}[htb]
  \caption{Feature approximation modification to Algorithm \ref{alg:algorithm1}.}
  \label{alg:algorithm2}
  \begin{algorithmic}[]
\State In Step \ref{alg1:step3} of Algorithm \ref{alg:algorithm1}, replace $T(X_j)$ with $T^*(X_j)$, where 
          \begin{equation}\label{method2:candidate}
          \begin{split}
          T^*(X_j)=\{&u\in X_N\setminus X_j:\|\Theta(u-v)\|_2\leq y \text{ and } \cos^2(\vartheta)\geq\delta_{j+1}/\|C_{X_j}(x)\|^2_2\\
          & \text{ for some } v\in \{x,X_j\}\},\nonumber
          \end{split}
          \end{equation}
and $\vartheta$, $C_{X_j}(x)$ are defined in Theorem \ref{featureThm}. Then,
          \begin{displaymath}
           x_{j+1}=\arg\max\limits_{u\in T^*(X_j)}{R(u)}.
          \end{displaymath}
  \end{algorithmic}
\end{algorithm*}

As an illustration of how cones in feature space relate to the design space, consider a full design $X_N$ consisting of 2500 ${\rm Unif}(0,1)$ data points, plotted in gray and yellow in the {\em left} panel of Figure \ref{f:reduced_feature}. The correlation function is $\Phi(x,x')={\rm exp}\{-\|(x-x')/10\|_2^2\}$ and the predictive location of interest is $x=(0.5,0.5)$, shown as a black triangle in the {\em left} panel. The first 7 design points are chosen greedily and indicated with red numbers. The {\em right} panel shows the first 2 components of the feature space (the first two eigenfunctions evaluated at the design points), colored and labeled correspondingly. The vector $C_{X_7}(x)$ is denoted as the middle dotted line in the {\em right} panel, with $|\vartheta|\le\pi/20$ shown as the outer dotted lines. Design points falling within these cones are shown in yellow in both panels. The design points in the {\em left} panel which fall in the yellow stripe have the most potential to reduce predictive variance.

\begin{figure}[ht!]
\centering
\includegraphics[angle=0,scale=0.65,trim=0 10 0 10]{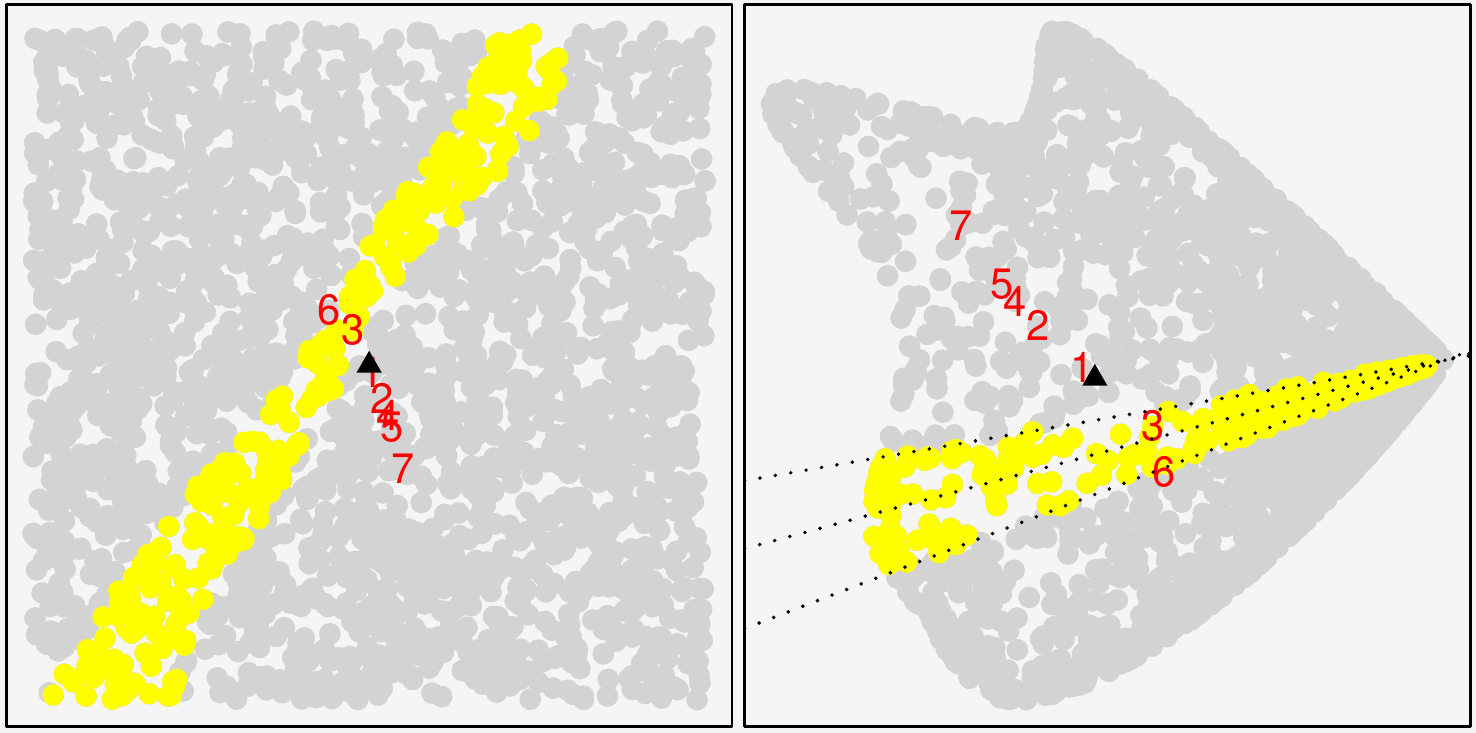}
\caption{\small Gray and yellow dots represent design points in the original space ({\em left}) 
and a $D=2$ dimensional feature space approximation ({\em right)}. Location of interest and current design  are annotated with black triangle and red numbers, respectively. Vector $C_{X_7}(x)$ and cones $|\vartheta|\le \pi/20$ shown with dotted lines. Design points falling within these cones are shown in yellow in both panels.}\label{f:reduced_feature}
\end{figure}

The computational complexity and storage of the proposed algorithms are summarized in Table \ref{table:algorithmcomparison}. Here, the original greedy approach proposed in \cite{gramacy2015local} is referred to as exhaustive search. Recall that $T(X_j)$ and $T^*(X_j)$ are the candidate sets from maximum distance method and feature approximation method, respectively. Let $|\cdot|$ denote the cardinality of a set. Since $|T(X_j)|$ and $|T^*(X_j)|$ are expected to be much smaller than $N$, the computational cost of the two proposed algorithms can be substantially reduced at each stage $j$ relative to the original greedy search.
However note that 
preprocessing time, for computing benchmarks and eliminating search locations, is required for both methods. 
Also, with a k-d tree or LSH search method, the specially adapted data structure indeed improves  computational efficiency during the preprocessing period ($O(j^2+jN)\rightarrow O(j^2+j\log(N))$ and $O(j^2+D^2N)\rightarrow O(j^2+D^2+jDN^{1/(1+\epsilon)})$, respectively). Considering the two proposed methods, $|T^*(X_j)|$ might be expected to be much smaller than $|T(X_j)|$ if (i) the correlation function is well approximated by the finite set of eigenfunctions and eigenvalues and (ii) the dimension of input is not too large, since distance becomes a very powerful exclusion criteria in even moderately high-dimensional space. On the other hand, the maximum distance method has 
smaller storage and preprocessing requirements. 
Section \ref{sec:example} presents two examples implementing the two proposed methods and shows the comparison. 

\begin{table}[h]
\centering
\scalebox{0.84}{
\begin{tabular}{|c|c||cc|cc|}
\toprule
 & Exhaustive &\multicolumn{2}{c}{Maximum Distance} & \multicolumn{2}{|c|}{Feature Approximation Method}\\
& Search & \multicolumn{2}{c}{Method} & \multicolumn{2}{|c|}{with $D$ Features*}\\
\midrule
{}   & {}   & w/o k-d tree   & w/ k-d tree   & w/o LSH   & w/ LSH\\
\midrule    
  Storage & $N$ & $N$ & $N$ &$ND$ & $ND$\\ 
  Preprocessing &  & $O(j^2+jN)$ & $O(j^2+j\log(N))$ & $O(j^2+D^2N)$ & $O(j^2+D^2+jDN^{1/(1+\epsilon)})$ \\    
  Search & $O(j^2N)$ & $O(j^2|T(X_j)|)$ & $O(j^2|T(X_j)|)$ & $O(j^2|T^*(X_j)|)$ & $O(j^2|T^*(X_j)|)$\\ 
\bottomrule
\end{tabular}
}
\caption{Complexity comparison between exhaustive search and two proposed methods for each update $j\rightarrow j+1$. The notation $|\cdot|$ denotes the cardinality of a set, and $\epsilon$ is a pre-specified value for the LSH method. 
*The complexity of pre-computation for feature approximation method is $O(D^3)$.}
\label{table:algorithmcomparison}
\end{table}

\section{Examples}\label{sec:example}
Two examples are discussed in this section: a two-dimensional example which demonstrates the algorithm and visually illustrates the reduction of candidates; and a larger-scale, higher-dimensional example. Both examples show the proposed methods considerably outperforming the original search method with respect to computation time. All numerical studies were conducted using \textsf{R} \citep{R} on a laptop with 2.4 GHz CPU and 8GB of RAM. The k-d tree and LSH were implemented via \textsf{R} package \texttt{RANN} \citep{RANN} and modifications to the source code of the \textsf{Python} package \texttt{scikit-learn} \citep{scikit-learn, bawa2005lsh}, and accessed in \textsf{R} through the \texttt{rPython} package \citep{rpython}.

\subsection{Two-dimensional problem of size $N=50^2$}\label{sec:example1}
Consider a computer experiment with full set of design locations $X_N$ consisting of a regular $50\times 50$ grid on $[-10,10]^2$ (2500 design points, light blue in Figure \ref{fig:illustration}) and take the predictive location of interest $x$ to be $(0.216,0.303)$ (circled $\times$ in Figure \ref{fig:illustration}). Set $\sigma^2=1$, and consider the Gaussian correlation function
\[\Phi_\Theta(x,y)=\exp\left\{-\left(\frac{(x_1-y_1)^2}{\theta_1}+\frac{(x_2-y_2)^2}{\theta_2}\right)\right\},\]
with $\theta_1=\theta_2=3$. This correlation function implies the $\phi$ in Algorithm \ref{alg:algorithm1} is $\phi(u)=\exp\{-u^2\}$ and $\Theta={\rm diag}(1/\sqrt{\theta_1},1/\sqrt{\theta_2})$. Then, we have $\phi^{-1}(v)=\sqrt{-\log{v}}$. 

Figure \ref{fig:illustration} illustrates the sub-design selection procedure shown in Algorithm \ref{alg:algorithm1}, in which $k=8$ nearest neighbors (from the candidate set) are used to generate the threshold in Step \ref{alg1:step2}. In Figure \ref{fig:illustration},
the gray dots represent the current design $X_j(x)$, the red dots represent the
optimal augmenting point $x_{j+1}$, and the points which are excluded from the
search for that location are those which fall outside the yellow shaded
region. The panels in the figure correspond to greedy search steps $j\in
\{3,16,29\}$. Notably, the optimal additional design points illustrated in Figure \ref{fig:illustration} are not always the nearest neighbors to the location of interest. In this example, only 7.40\% (185/2500) of candidates need to be
searched in the beginning. Even after choosing thirty data points, there is no
need to search much more than half of the full data (56.92\%=1423/2500).

Continuing the same example, Figure \ref{fig:illustration} also shows substantial improvement from the feature approximation method. In the example, a $D=500$ dimensional feature space approximation is pre-computed using Nystr{\"o}m's method \citep{williams2001using}. The points annotated with green $+$s are the points which are not excluded from the search. In fact, the number of candidates which need to be searched is usually reduced at least 10 fold and in many cases 50 or 100 fold, or more. 

\begin{figure}[ht!]
\centering
\vspace{-0.1cm}
\includegraphics[scale=0.31,trim=5 20 5 55,clip=TRUE]{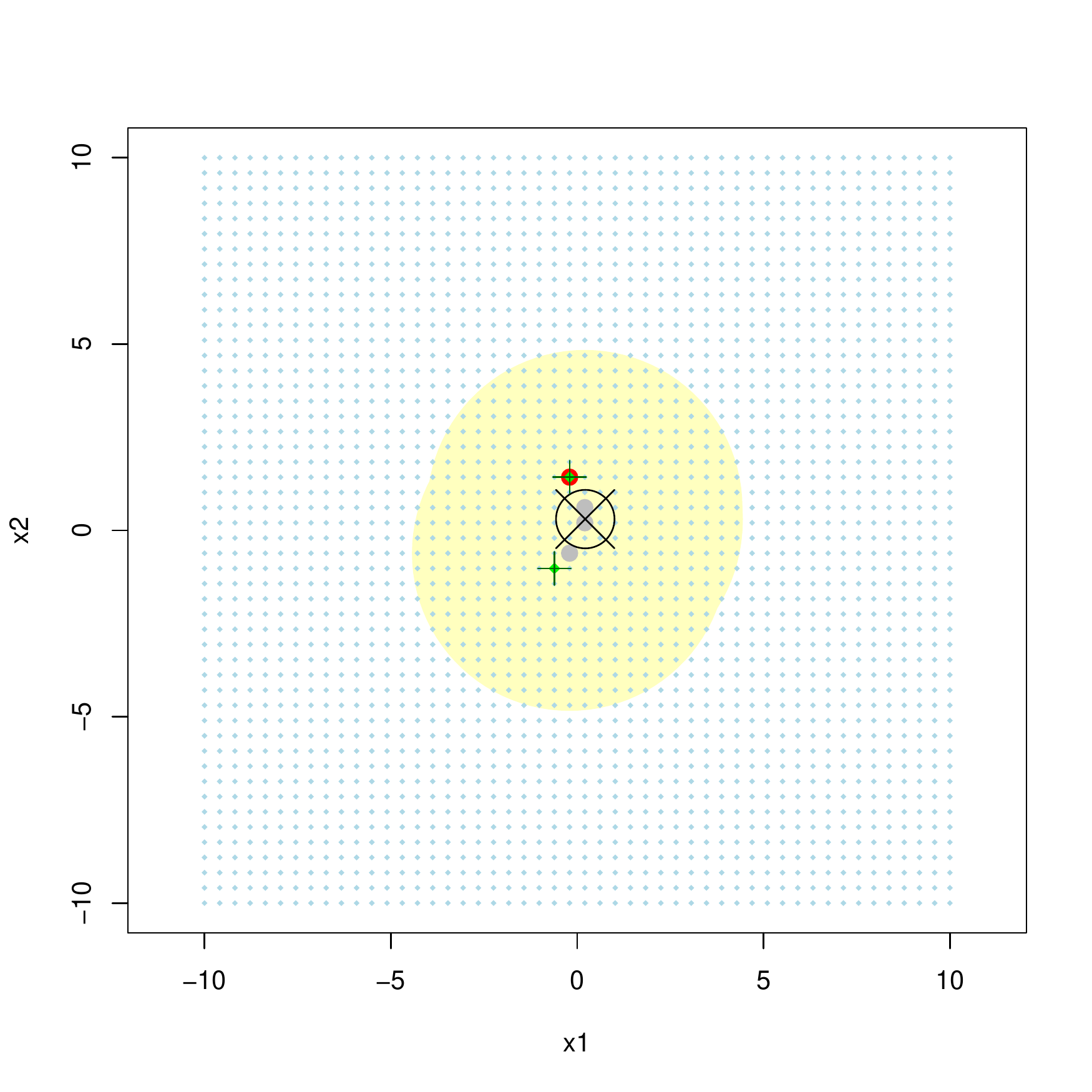}
\includegraphics[scale=0.31,trim=5 20 5 55,clip=TRUE]{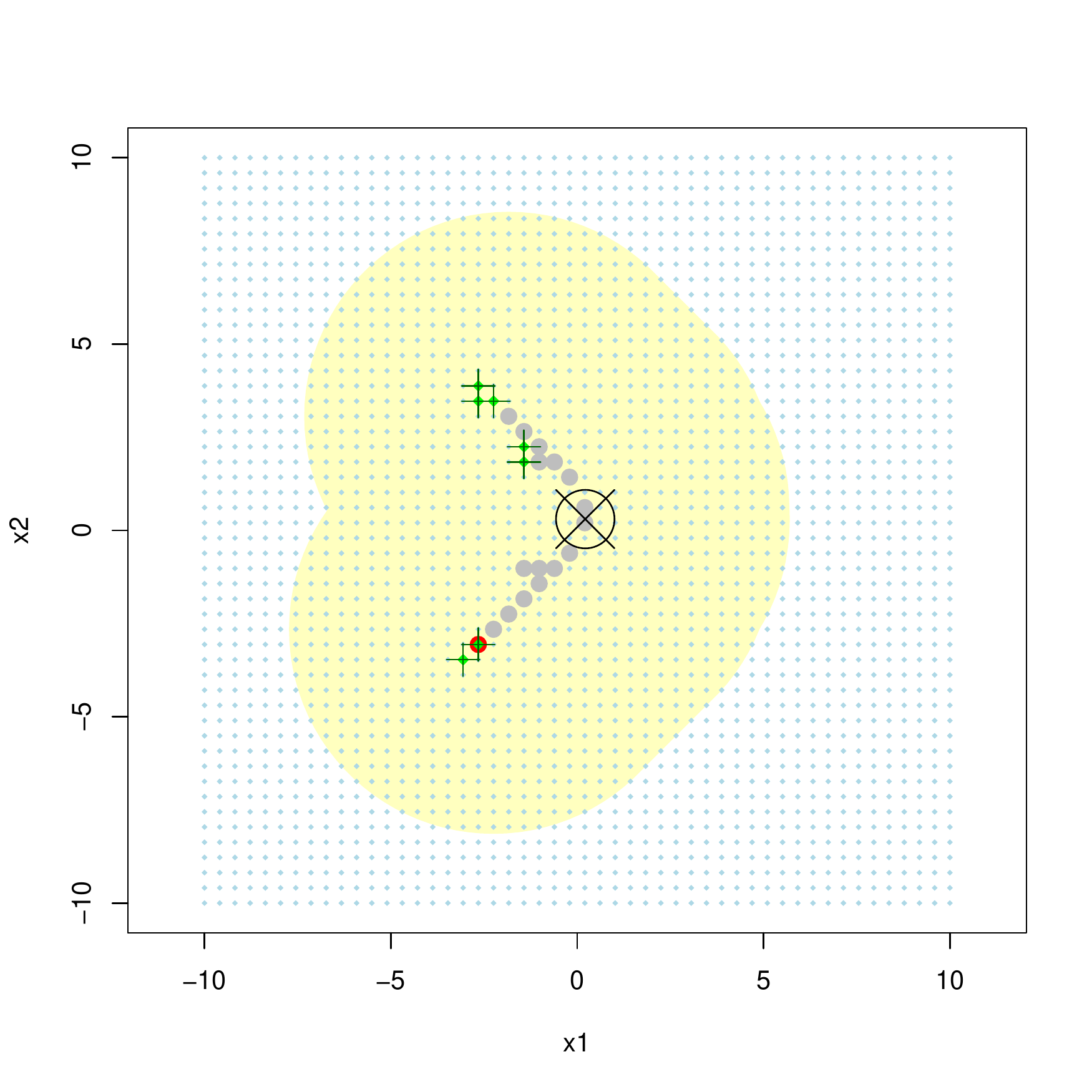}
\includegraphics[scale=0.31,trim=5 20 5 55,clip=TRUE]{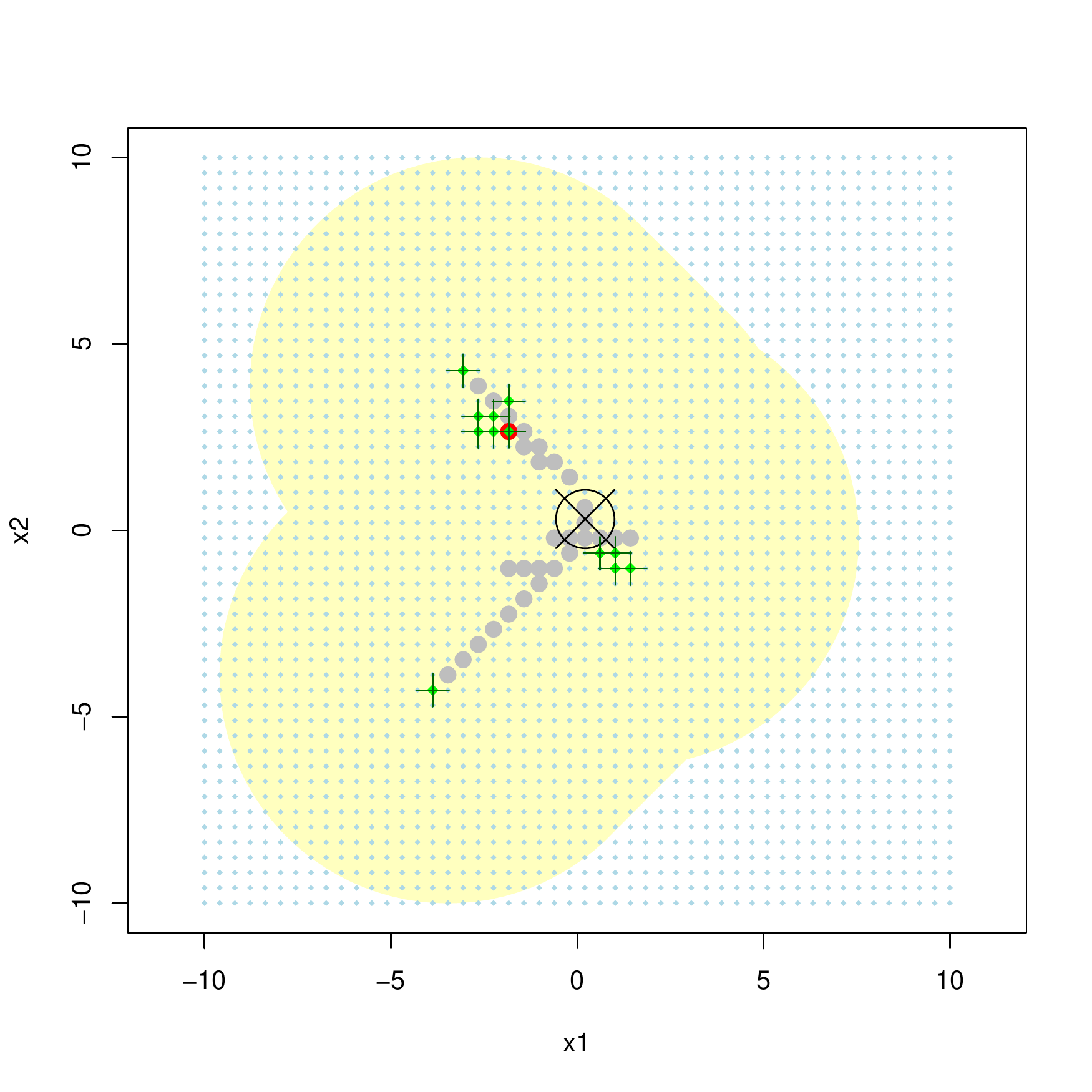}
\vspace{-0.65cm}
\caption{\small Left, middle, and right panels respectively illustrate  selection
at $j=3,16,$ and $29$. The circled $\times$ is the location of interest, $(0.216,0.303)$. Gray dots are the current design points;
red dots are the optimal $x_{j+1}$; points which are
excluded from the search based on maximum distance method are 
those which fall outside the yellow shaded region. Points which are
excluded from the search based on feature approximation method are 
those which are not annotated with a green $+$.}\label{fig:illustration}
\end{figure}

%


While the maximum distance method and original greedy approach proposed in \cite{gramacy2015local} produce the same sub-designs and in turn the same predictive variances,
the feature approximation method is \emph{approximate} and can produce different sub-designs and in turn slightly different predictive variances (not \emph{necessarily} inflated due to greedy nature of search).
Table \ref{table:variancecomparison} shows relative differences in predictive variance resulting from feature approximation method with $D=10$, $200$ and $500$ features as compared to maximum distance method (or equivalently the original greedy approach).
The number of search candidates is listed in parentheses.
The relative difference in predictive variance 
is defined as
\[\frac{V_{j,FA}(x)-V_{j,MD}(x)}{V_{j,MD}(x)},\]
where $V_{j,FA}(x)$ and $V_{j,MD}(x)$ denote the predictive variance of the emulator at location $x$ at stage $j$ using the feature approximation method and maximum distance method, respectively. 
As might be expected, a larger number of features, $D=500$, reduces the search candidates without any loss in variance reduction. For $D = 200$, although there are small differences in predictive variance, the discrepancies may be small enough to be of little practical consequence. At stage 15, 20, and 25, the predictive variance for $D=200$ is even smaller than maximum distance method, due to the greedy nature of the searches. 
Notably, if a small number of features, say $D=10$, is chosen, feature approximation search offers little improvement over maximum distance method in terms of search reduction, even though the predictive variances are similar to maximum distance method. 
In this case, $D=200$ features might be a reasonable choice, balancing ease of computation and small predictive variance. 
\begin{table}[h]
\centering
\begin{tabular}{|c|cc|cc|C{0.8cm}c||c|}
\toprule 
Relative Difference & & & & & & & Variance by\\
(\# of searching  &\multicolumn{2}{c|}{$D=10$} &\multicolumn{2}{c|}{$D=200$} & \multicolumn{2}{c||}{$D=500$} & Maximum\\
    candidates) & & &  & &  &  &  Distance Method\\
\midrule 
  Stage 10 & 0 & (842) & 0.178 & (4) & 0 & (47) & $1.95\times 10^{-6}$ (844)\\ 
  Stage 15 & 0.006 & (1057)  & -0.76 & (7) & 0 & (69) & $9.35\times 10^{-7}$ (1040)  \\ 
  Stage 20 & 0.018 & (1149)  & -0.722  & (30) & 0 & (12) & $6.12\times 10^{-7}$ (1168)  \\ 
  Stage 25 & -0.155 & (1332)  & -0.091 & (4) & 0 & (116) & $1.66\times 10^{-7}$ (1295)\\ 
  Stage 30 & 0.009 & (1459) & 0.024 & (20) & 0 & (2) & $1.28\times 10^{-8}$ (1423) \\
\bottomrule
\end{tabular} 
\caption{The relative difference in variance of the emulator at location $(0.216,0.303)$ between maximum distance search as a baseline and feature approximation search with number of features $D$: 10, 200 and 500. 
Baseline variance by maximum distance search is shown in the last column. The value in parentheses is the number of search candidates.}
\label{table:variancecomparison}
\end{table}

To further compare the performance of the proposed methods with original greedy approach (exhaustive search), a Sobol's quasi-random sequence \citep{bratley1988algorithm} of 100 predictive locations is generated. 
Table \ref{table:totalcomparisonEX1} shows the average computation time and proportion of search candidates for the proposed methods and exhaustive search over the 100 predictive locations. 
The 
proportion
of search candidates for the maximum distance method is from 22.78\% to 39.15\%. The method also marginally speeds up computation time from 21 to 18 seconds with a k-d tree data structure. 
On the other hand, although the feature approximation method needs 6 seconds for computing the features in advance, the proportion of search candidates for feature approximation method with $D=200$ is reduced to 5.88\% at stage 30. The computation time, on an ordinary laptop, is less than 15 seconds for 30 stages of iteration in the $N=50^2$ experiment. 
Relative average predictive variance increases due to using the feature approximation method, both with and without LSH, are shown in Table \ref{table:variancecomparisonEX1}.
At stage 30 the average predictive variance increases due to using the feature approximation method are, with and without LSH, 4.6\% and 1.4\%, respectively, potentially small enough to be disregarded in a practical context.
The LSH data structure
also marginally reduces search time from 13 to 11 seconds. Recall that the feature approximation method with LSH approximates both the covariance function and the cosine similarity measure, so the candidate set is slightly different from the one without LSH.
While
in this moderately-sized problem the k-d tree and LSH data structures do not greatly improve the computational cost (at stage 30, k-d tree: $21\rightarrow 18$, LSH: $13\rightarrow 11$), in a larger-scale problem the improvements due to incorporating a k-d tree or LSH data structure can be relatively substantial, as will be shown in next subsection.

\begin{table}[h]
\centering
\begin{tabular}{|c|c||cc|cc|}
\toprule
Seconds & Exhaustive &\multicolumn{2}{c}{Maximum Distance} & \multicolumn{2}{|c|}{*Feature Approximation}\\
(Candidates \%) & Search & \multicolumn{2}{c}{Method} & \multicolumn{2}{|c|}{Method with $D=200$}\\
\midrule
{}   & {}   & w/o KD-tree   & w/ KD-tree   & w/o LSH   & w/ LSH\\
\midrule
  Stage 10 & 11 & 3 (22.78\%)& 2 (22.78\%) & 3 (2.69\%)& 2 (1.98\%)\\      
  Stage 15 & 19 & 5 (28.04\%) & 4 (28.04\%)& 5 (3.31\%)& 4 (2.90\%)\\ 
  Stage 20 & 30 &  9 (32.68\%)& 7 (32.68\%)& 8 (7.57\%) & 6 (5.56\%) \\ 
  Stage 25 & 44 & 14 (36.10\%) &12 (36.10\%)&10 (6.41\%) &  8 (5.66\%)\\ 
  Stage 30 &  61 & 21 (39.15\%) & 18 (39.15\%) & 13 (5.88\%) & 11 (6.65\%)\\
\bottomrule
\end{tabular}
\caption{Average time (seconds) comparison between exhaustive search and two proposed methods in two-dimensional setting with $N=50^2$ over 100 Sobol predictive locations. The values in parentheses are the average percentage searched of full design. *Pre-computation time for feature approximation method is 6 seconds.}
\label{table:totalcomparisonEX1}
\end{table}

\begin{table}[h]
\centering
\begin{tabular}{|c||cc||c|}
\toprule
Relative Difference  &\multicolumn{2}{c||}{Feature Approximation} & Average Variance by\\
 & \multicolumn{2}{c||}{Method with $D=200$} &  Maximum Distance Method\\
\midrule
{}   & w/o LSH   & w/ LSH & {}\\
\midrule
  Stage 10 & 0.192 & 0.168 & $4.15\times 10^{-6}$ \\      
  Stage 15 & 0.348 & 0.140 & $1.08\times 10^{-6}$ \\ 
  Stage 20 & 0.171 & 0.066 & $4.34\times 10^{-7}$ \\ 
  Stage 25 & 0.011 & -0.124 & $2.26\times 10^{-7}$\\ 
  Stage 30 & 0.046 & 0.014  & $1.62\times 10^{-7}$\\
\bottomrule
\end{tabular}
\caption{The relative difference in average predictive variance of the emulator between maximum distance search as a baseline and feature approximation search with number of features $D=200$ over 100 Sobol predictive locations in $2$-dimensional setting.}
\label{table:variancecomparisonEX1}
\end{table}

\subsection{$6$-dimensional problem of size $N=5\times 10^4$}\label{sec:example2}
%

Even more substantial reductions in the number of search candidates are seen for both methods in a larger-scale, higher-dimensional setting. In this example, we generate a $6$-dimensional Sobol's quasi-random sequence of size $N=5\times 10^4$ in a $[-1,1]^{6}$ for the design space and the predictive locations are chosen from a Sobol's quasi-random sequence of size 20. Set $\sigma^2=1$ and tuning parameter $k=30$, and take the correlation function $\Phi_\Theta(x,y)=\exp\{-\sum^{6}_{i=1}\frac{(x_i-y_i)^2}{\theta_i}\}$ with $\theta_i=1.5,i=1,\ldots,6$.

Table \ref{table:totalcomparison} shows the comparison between exhaustive search and the two proposed methods. 
As the table shows, the two proposed methods outperform exhaustive search in terms of computation time. 
Further, the number of candidates searched for both methods are less than $10\%\;(=5000/50000)$ across all 30 stages. While the exhaustive search takes 3423 seconds ($\approx$ 1 hours) for 30 stage iterations, 240 seconds (4 minutes) are required for maximum distance method. Incorporating a k-d tree data structure, the computation time decreases to 193 seconds ($\approx$ 3.2 minutes). Compared to the 2-dimensional example in Section \ref{sec:example1}, incorporating a k-d tree data structure has moderately more computational benefit in this larger-scale setting.

The feature approximation method, as expected, has a smaller-sized candidate set than maximum distance method. Moreover, using $D=300$ features, less than 2\% average predictive variance increases at stage 30 are observed due to approximation, as shown in Table \ref{table:variancecomparisonEX2}. 
On the other hand, due to the moderately expensive computation in Algorithm \ref{alg:algorithm2} using $D=300$ features, in this example feature approximation search without LSH is more time-consuming than the maximum distance method. 
As shown in Table \ref{table:algorithmcomparison}, the computation of more design points incurs higher computational costs in order of $D^2$ for feature approximation search without LSH (complexity $O(j^2+D^2N)$). With an LSH approximate similarity-search method, computation time is reduced by 189 seconds ($\approx$ 3 minutes) across all 30 stages.
While the feature approximation approach outperforms exhaustive search,
it appears to be most useful when 
the maximum distance approach is very conservative, such as in the two-dimensional case in Section \ref{sec:example1}.

\begin{table}[h]
\centering
\begin{tabular}{|c|c||cc|cc|}
\toprule
Seconds & Exhaustive &\multicolumn{2}{c}{Maximum Distance} & \multicolumn{2}{|c|}{*Feature Approximation}\\
(Candidates \%) & Search & \multicolumn{2}{c}{Method} & \multicolumn{2}{|c|}{Method with $D=300$}\\
\midrule
{}   & {}   & w/o KD-tree   & w/ KD-tree   & w/o LSH   & w/ LSH\\
\midrule
  Stage 10 & 488 & 24 (2.77\%)& 10 (2.77\%)& 74 (1.71\%) & 26 (1.7\%)\\      
  Stage 15 & 953 & 50 (4.27\%) & 28 (4.27\%) & 126 (3.34\%)& 45 (3.68\%)\\ 
  Stage 20 & 1601 & 93 (5.84\%)& 62 (5.84\%) & 199 (4.77\%) & 76 (5.16\%)\\ 
  Stage 25 & 2423 & 154 (7.34\%) & 115 (7.34\%)& 296 (5.28\%) & 121 (5.21\%)\\ 
  Stage 30 & 3423 & 240 (8.62\%) & 193 (8.62\%) & 435 (6.70\%)& 189 (6.38\%)\\
\bottomrule
\end{tabular}
\caption{Time (seconds) comparison between exhaustive search and two proposed methods in $6$-dimensional setting with $N=5\times 10^4$ over 20 Sobol predictive locations. The values in parentheses shows the percentage searched of full design. *Pre-computation time for feature approximation method is 26 seconds.}
\label{table:totalcomparison}
\end{table}

\begin{table}[h]
\centering
\begin{tabular}{|c||cc||c|}
\toprule
Relative Difference  &\multicolumn{2}{|c|}{Feature Approximation} & Average Variance by\\
 & \multicolumn{2}{|c|}{Method with $D=300$} &  Maximum Distance Method\\
\midrule
{}   & w/o LSH   & w/ LSH & {}\\
\midrule
  Stage 10 & 0.049 & 0.047 & $0.2328$ \\      
  Stage 15 & 0.030 & 0.032 & $0.2120$ \\ 
  Stage 20 & 0.023 & 0.022 & $0.1997$ \\ 
  Stage 25 & 0.017 & 0.017 & $0.1913$\\ 
  Stage 30 & 0.016 & 0.016  & $0.1850$\\
\bottomrule
\end{tabular}
\caption{The relative difference in average predictive variance of the emulator between maximum distance search as a baseline and feature approximation search with number of features $D=300$ over 20 Sobol predictive locations in $6$-dimensional setting.}
\label{table:variancecomparisonEX2}
\end{table}


%

\section{Conclusion and Discussion}\label{sec:conclusion}
Emulators have become crucial for approximating the relationship between input and output in computer simulations. However, as data sizes continue to grow, GP emulators fail to perform well due to memory, computation, and numerical issues. In order to deal with these issues, \cite{gramacy2015local} proposed a local GP emulation technique accompanied by a sequential scheme for building local sub-designs by maximizing reduction in variance.  We showed that an important (exhaustive) search subroutine could be substantially shortcut without compromising on accuracy, leading to substantial reductions in computing time.

In particular, 
using the distance-based structure of most correlation functions in GP models, 
we showed that input locations distant from the predictive location of interest offer little potential for variance reduction.
We proposed a \textit{maximum distance method} to speed up construction of local GP emulators on the neighborhood of the existing sub-design and predictive location. 
Taking a step further, we observed that, since the correlation functions in GP models can be uniformly approximated by a finite sum of features via eigen-decomposition, 
mapping the original space into a feature space by the eigenfunctions can further reduce the search scope. We developed a \textit{feature approximation method} that determines viable candidates in terms of the angle between two projected feature vectors.
This leads to an even smaller proportion of viable candidates for searching. Taken together,
the two reductions lead to an order of magnitude smaller search set. 

We provided two examples that illustrate how the two search methods
perform.  
Obtaining accurate predictions for large-scale problems takes only a few minutes, on an ordinary laptop.  For instance, maximum distance search leveraging a k-d tree data structure takes less than 4 minutes to search for effective candidates in the second example while the full search, by comparison, takes about one hour.

The two proposed methods can be extended for selecting more than one point in each stage $j$ in a straight-forward manner. 
For example, suppose two points are to be selected in each stage. 
Let $j'=2j$, $X_{j'}$ be the current sub-design at stage $j$, and $x_{j'+1}$ and $x_{j'+2}$ be two points selected at the stage $j+1$. 
Proposition 1 can be extended to $V_{j+1}(x)=V_j(x)-\sigma^2R^*(x_{j'+1},x_{j'+2})$ for a function $R^*$,
Theorem 1 can narrow the window of potential \emph{pairs} of candidate locations, to say $T'(X_j)$, and Algorithm \ref{alg:algorithm1} can be updated accordingly.
On the other hand, retaining good computational properties in a batch-sequential framework is not straight-forward. 
For example, searching for the optimal candidates, $(x_{j'+1},x_{j'+2})=\arg\max_{(u_1,u_2)\in T'(X_j)}R^*(u_1,u_2)$, might be very expensive, say $O(|T'(X_j)|^2)$, compared to searching for one point in each stage.
Efficiently augmenting multiple points at each stage, for example by alternating maximizations on $x_{j'+1}$ and $x_{j'+2}$, might be worth exploring in future work.

The essential ideas of the proposed approaches have potential for application 
in search space reduction in global optimization.
Consider the following example.
\cite{chen2008sequential} modified the \textit{maximum entropy design} \citep{shewry1987maximum} for use as a sequential algorithm
to efficiently construct a space-filling design in computer experiments.
They showed that the algorithm can be simplified to selecting a new point that maximizes the so-called \textit{sequential maximum entropy criterion}
\[
x_{j+1}=\arg\min\limits_{u\in \mathcal{D}\setminus X_j}{\Phi_\Theta(u,X_j)\Phi_\Theta(X_j,X_j)^{-1}\Phi_\Theta(X_j,u)},
\]
where $\mathcal{D}$ is a \textit{discrete} design space. 
For this global optimization problem, 
let
\begin{align*}
d_{\max}(x_{j+1})=\max\{\|\Theta(x_1-x_{j+1})\|_2,\|\Theta(x_2-x_{j+1})\|_2,\ldots,\|\Theta(x_j-x_{j+1})\|_2\}
\end{align*}
and $\delta>0$.
It can be shown that if $d_{\max}(x_{j+1})\leq\phi^{-1}\left(\sqrt{\lambda_{\max}\delta}\right)$, where $\lambda_{\max}$ is the maximum eigenvalue of $\Phi_\Theta(X_j,X_j)$, then the objective function $\Phi_\Theta(x_{j+1},X_j)\Phi_\Theta(X_j,X_j)^{-1}\Phi_\Theta(X_j,x_{j+1})>\delta$. 
Thus, similar to Algorithm \ref{alg:algorithm1}, a maximum distance approach could be used to eliminate search candidates. 
For other specific global optimization problems, detailed examination is needed.

An implicit disadvantage of these methods is the impact of the correlation parameters. Take the example in Figure \ref{mindistance}, where $d_{\min}(x_{9})<3.07$. From the definition \eqref{mindisdef} of $d_{\min}(x_{j+1})$, suppose $\Theta=(1/\sqrt{\theta},1/\sqrt{\theta})$, then the larger $\theta$ is, the bigger the search area, the yellow shaded region in Figure \ref{mindistance}. The reason is that when $\theta$ is large, the correlation is close to one and the data points tend to be highly correlated, implying that every data point in the full design carries important information for each predictive location. In other words,
the algorithm requires more computation for ``easier'' problems---i.e., with a ``flatter'' surfaces.  
On the other hand ``flatter'' surfaces do not require large sub-designs to achieve small predictive variance.

An improvement worth exploring is how to determine of the number of features $D$ in the feature approximation method. Cross-validation to minimize predictive variance of an emulator may present an attractive option. 
Finally, a examination of the choice between the maximum distance and feature approximation methods might be desirable. 
Although using them both in concert guarantees a smaller candidate set in the feature approximation method, pre-computation of the features constitutes a moderately expensive sunk cost in terms of computation and storage. 
In the example in Section \ref{sec:example2}, a $500\times 50,000$ matrix needed to be computed and stored in advance. In either case, the two methods outperform exhaustive search as shown in Table \ref{table:totalcomparison}.


\subsection*{Acknowledgments}

All three authors would like to gratefully acknowledge funding from the
National Science Foundation, award DMS-1621746.

\section{Appendices}
\subsection{Proof of Proposition \ref{prop1}}\label{proofreducevariance}
In the variance definition \eqref{varfun}, the variance of $Y(x)$ at stage $j+1$ is
\begin{equation}\label{varfunnewloaction}
V_{j+1}(x)=\sigma^2\{\Phi_\Theta(x,x)-\Phi_\Theta(x,X_{j+1})\Phi_\Theta(X_{j+1},X_{j+1})^{-1}\Phi(X_{j+1},x)\}.
\end{equation}

Since $X_{j+1}$ is comprised of $X_j$ and $x_{j+1}$, \eqref{varfunnewloaction} can be rewritten as
\begin{multline}\label{beforesimpliedpartition}
V_{j+1}(x)=\sigma^2\Big\{\Phi_\Theta(x,x)-\\
\;\;\; \;\;\; \left[\begin{matrix}\Phi_\Theta(x,x_{j+1})&\Phi_\Theta(x,X_j)\end{matrix}\right]
\left[\begin{matrix}
\Phi_\Theta(x_{j+1},x_{j+1})  & \Phi_\Theta(x_{j+1},X_{j})\\ 
\Phi_\Theta(X_{j},x_{j+1}) & \Phi_\Theta(X_{j},X_{j})
\end{matrix}\right]^{-1}
\left[\begin{matrix}
\Phi_\Theta(x,x_{j+1}) \\ 
\Phi_\Theta(X_{j},x)
\end{matrix}\right]\Big\}.
\end{multline}

For simplicity, the second term of \eqref{beforesimpliedpartition} can be written as a partitioned matrix, that is,
\begin{equation}\label{aftersimpliedpartition}
\left[\begin{matrix}a^T_1 & a^T_2\end{matrix}\right]
\left[\begin{matrix}
B_{11} &  B_{12} \\
B_{21} & B_{22}
\end{matrix}\right]^{-1}
\left[\begin{matrix}a_1 \\ a_2\end{matrix}\right],
\end{equation}
where 
\begin{equation*}
a_1=\Phi_\Theta(x,x_{j+1}), a_2=\Phi_\Theta(X_{j},x),
\end{equation*}
\begin{equation*}
B_{11}=\Phi_\Theta(x_{j+1},x_{j+1}), B_{12}=\Phi_\Theta(x_{j+1},X_{j})=B^T_{12} \text{ and } B_{22}=\Phi_\Theta(X_{j},X_{j}).
\end{equation*}

Applying partitioned matrix inverse results \citep{harville1997matrix} and simplifying \eqref{aftersimpliedpartition} gives
\begin{equation}\label{partitionresult}
a^T_2B^{-1}_{22}a_2+(a_1-B_{12}B^{-1}_{22}a_2)^TB^{-1}_{11\cdot 2}(a_1-B_{12}B^{-1}_{22}a_2),
\end{equation}
where $B_{11\cdot 2}=B_{11}-B_{12}B^{-1}_{22}B_{21}$.

Then, taking \eqref{partitionresult} into \eqref{beforesimpliedpartition} leads to 
\begin{align*}
V(X_{j+1})&=\sigma^2\{\Phi_\Theta(x,x)-a^T_2B^{-1}_{22}a_2-(a_1-B_{12}B^{-1}_{22}a_2)^TB^{-1}_{11\cdot 2}(a_1-B_{12}B^{-1}_{22}a_2)\}\\
&=\sigma^2\{\Phi_\Theta(x,x)-\Phi_\Theta(x,X_j)\Phi_\Theta(X_j,X_j)^{-1}\Phi_\Theta(X_j,x)\\
&\;\;\;\;\; -(a_1-B_{12}B^{-1}_{22}a_2)^TB^{-1}_{11\cdot 2}(a_1-B_{12}B^{-1}_{22}a_2)\}\\
&=V(X_{j})-\sigma^2\{(a_1-B_{12}B^{-1}_{22}a_2)^TB^{-1}_{11\cdot 2}(a_1-B_{12}B^{-1}_{22}a_2)\}\\
&=V(X_{j})-\sigma^2R(x_{j+1}),
\end{align*}
where 
\begin{align*}
R(x_{j+1})&=(a_1-B_{12}B^{-1}_{22}a_2)^TB^{-1}_{11\cdot 2}(a_1-B_{12}B^{-1}_{22}a_2)\\
&=(a_1-B_{12}B^{-1}_{22}a_2)^2/B_{11\cdot 2}\\
&=\frac{(\Phi_\Theta(x,x_{j+1})-\Phi_\Theta(x_{j+1},X_j)\Phi_\Theta(X_j,X_j)^{-1}\Phi_\Theta(X_j,x))^2}{\Phi_\Theta(x_{j+1},x_{j+1})-\Phi_\Theta(x_{j+1},X_j)\Phi_\Theta(X_j,X_j)^{-1}\Phi_\Theta(X_j,x_{j+1})},
\end{align*}
and the second equality holds since $B_{11\cdot 2}$ is a scalar.

\subsection{Proof of Theorem \ref{theorem}}\label{proofthm1}
Since $(a-b)^2\leq (a+b)^2$ for $a,b\geq 0$, equation \eqref{inequation} can be bounded as
\begin{align*}
R(x_{j+1})&=\frac{(\Phi_\Theta(x,x_{j+1})-\Phi_\Theta(x_{j+1},X_j)\Phi_\Theta(X_j,X_j)^{-1}\Phi_\Theta(X_j,x))^2}{\Phi_\Theta(x_{j+1},x_{j+1})-\Phi_\Theta(x_{j+1},X_j)\Phi_\Theta(X_j,X_j)^{-1}\Phi_\Theta(X_j,x_{j+1})}\\
&\leq \frac{(\Phi_\Theta(x,x_{n+1})+\Phi_\Theta(x_{j+1},X_j)\Phi_\Theta(X_j,X_j)^{-1}\Phi_\Theta(X_j,x))^2}{\Phi_\Theta(x_{j+1},x_{j+1})-\Phi_\Theta(x_{j+1},X_j)\Phi_\Theta(X_j,X_j)^{-1}\Phi_\Theta(X_j,x_{j+1})}.
\end{align*}

Also, since \[a^TB^{-1}b\leq\|a\|_2\|B^{-1}b\|_2\] and 
\[a^TB^{-1}a\leq \|a\|^2_2\lambda_{\max}(B^{-1})=\|a\|^2_2/\lambda_{\min}(B),\]
where $\lambda_{\max}(\cdot)$ and $\lambda_{\min}(\cdot)$ denote the maximum and minimum eigenvalues of a specific matrix, respectively, the inequality becomes

\begin{align*}
R(x_{j+1})&\leq \frac{(\Phi_\Theta(x,x_{j+1})+\|\Phi_\Theta(x_{j+1},X_j)\|_2\|\Phi_\Theta(X_j,X_j)^{-1}\Phi_\Theta(X_j,x)\|_2)^2}{1-\|\Phi_\Theta(X_j,x_{j+1})\|^2_2/\lambda_{\min}},		 
\end{align*}
where $\lambda_{\min}$ is the minimum eigenvalue of $\Phi_\Theta(X_j,X_j)$. 

Furthermore, according to the definition $d_{\min}(x_{j+1})$ of the minimum (Mahalanobis-like) distance as \eqref{mindisdef} and the definition $\phi(\cdot)$ as in Theorem \ref{theorem}, we have \[\Phi_\Theta(u,x_{j+1})\leq\phi(d_{\min}(x_{j+1})), \text{ for any } u\in\{x, X_j\},\]
which also implies 
\[
\|\Phi_\Theta(x_{j+1},X_j)\|_2=\|\Phi_\Theta(X_j,x_{j+1})\|_2\leq \sqrt{j}\phi(d_{\min}(x_{j+1})),
\]
therefore the inequality can be bounded as 
\begin{equation}\label{finalinequ}
R(x_{j+1})\leq \frac{(\phi(d_{\min}(x_{j+1}))+\sqrt{j}\phi(d_{\min}(x_{j+1}))\|\Phi_\Theta(X_j,X_j)^{-1}\Phi_\Theta(X_j,x)\|_2)^2}{1-j\phi^2(d_{\min}(x_{j+1}))/\lambda_{\min}}.
\end{equation}

Thus, for $\delta>0$, if 
\[\frac{(\phi(d_{\min}(x_{j+1}))+\sqrt{j}\phi(d_{\min}(x_{j+1}))\|\Phi_\Theta(X_j,X_j)^{-1}\Phi_\Theta(X_j,x)\|_2)^2}{1-j\phi^2(d_{\min}(x_{j+1}))/\lambda_{\min}}\leq \delta\] 
or equivalently
\[d_{\min}(x_{j+1})\geq\phi^{-1}\left(\sqrt{\frac{\delta}{(1+\sqrt{j}\|\Phi_\Theta(X_j,X_j)^{-1}\Phi_\Theta(X_j,x)\|_2)^2+j\delta/\lambda_{\min}}}\right),\] then by \eqref{finalinequ}, $R(x_{j+1})\leq\delta$.

\subsection{Proof of Theorem \ref{featureThm}}\label{proof:reducevariancebyfeature}
Define $U(t)=(\sqrt{\lambda_1}\phi_1(t),\sqrt{\lambda_2}\phi_2(t),\ldots,\sqrt{\lambda_D}\phi_D(t))^T\in\mathbb{R}^{D\times 1}$, where $\phi_i(\cdot),i=1,\ldots,D$ is an orthonormal basis of $L^2(\Omega)$ consisting of the eigenfunctions of $T$, defined in \eqref{defineT}, and $\lambda_1\geq\lambda_2\geq\ldots\geq\lambda_D$ are corresponding eigenvalues. According to \eqref{approx_eigen_decomp}, the approximated eigen-decomposition can be rewritten as 
\[
\Phi(x,y) \approx U^T(x)U(y).
\]

Also, define a matrix $U(K)=[U(k_1),U(k_2),\ldots,U(k_n)]\in\mathbb{R}^{D\times n}$ for $K = (k_1,k_2,\ldots,k_n)$. Then, the reduction in variance $R(x_{j+1})$ in \eqref{reducevariance} can be approximated to the following:

\begin{align*}
R(x_{j+1})&=\frac{(\Phi_\Theta(x,x_{j+1})-\Phi_\Theta(x_{j+1},X_j)\Phi_\Theta(X_j,X_j)^{-1}\Phi_\Theta(X_j,x))^2}{\Phi_\Theta(x_{j+1},x_{j+1})-\Phi_\Theta(x_{j+1},X_j)\Phi_\Theta(X_j,X_j)^{-1}\Phi_\Theta(X_j,x_{j+1})}\\
&\approx\frac{(U_\Theta^T(x_{j+1})U_\Theta(x)-U_\Theta^T(x_{j+1})U_\Theta(X_j)[U_\Theta^T(X_j)U_\Theta(X_j)]^{-}U_\Theta^T(X_j)U_\Theta(x))^2}{U_\Theta^T(x_{j+1})U_\Theta(x_{j+1})-U_\Theta^T(x_{j+1})U_\Theta(X_j)[U_\Theta^T(X_j)U_\Theta(X_j)]^{-}U_\Theta^T(X_j)U_\Theta(x_{j+1})}\\
	&=\frac{\{U_\Theta^T(x_{j+1})[I-U_\Theta(X_j)[U_\Theta^T(X_j)U_\Theta(X_j)]^{-}U_\Theta^T(X_j)]U_\Theta(x)\}^2}{U_\Theta^T(x_{j+1})[I-U_\Theta(X_j)[U_\Theta^T(X_j)U_\Theta(X_j)]^{-}U_\Theta^T(X_j)]U_\Theta(x_{j+1})},
\end{align*}
where $[U_\Theta^T(X_j)U_\Theta(X_j)]^{-}$ denotes a generalized inverse of $[U_\Theta^T(X_j)U_\Theta(X_j)]$.
	
Let $C_{\Theta,X_j}(t)=[I-U_\Theta(X_j)[U_\Theta^T(X_j)U_\Theta(X_j)]^{-}U_\Theta^T(X_j)]U_\Theta(t)$. Then,
\begin{align*}
&C^T_{\Theta,X_j}(x_{j+1})C_{\Theta,X_j}(x)\\
=&U_\Theta^T(x_{j+1})[I-U_\Theta(X_j)[U_\Theta^T(X_j)U_\Theta(X_j)]^{-}U_\Theta^T(X_j)]U_\Theta(x).
\end{align*}
Similarly,
\[
C^T_{\Theta,X_j}(x_{j+1})C_{\Theta,X_j}(x_{j+1})=U_\Theta^T(x_{j+1})[I-U_\Theta(X_j)[U_\Theta^T(X_j)U_\Theta(X_j)]^{-}U_\Theta^T(X_j)]U_\Theta(x_{j+1}).
\]
Therefore,
\begin{displaymath}
R(x_{n+1})\approx\frac{(C_{\Theta,X_j}^T(x_{j+1})C_{\Theta,X_j}(x))^2}{C_{\Theta,X_j}^T(x_{j+1})C_{\Theta,X_j}(x_{j+1})}=\|C_{\Theta,X_j}(x)\|^2_2\cos^2(\vartheta),
\end{displaymath}
where $\vartheta$ is the angle between $C_{\Theta,X_j}(x)$ and $C_{\Theta,X_j}(x_{j+1})$.
\bibliography{bib}
\end{document}